\documentclass{article}

\newif\ifdraftversion

\draftversionfalse
\ifdraftversion
\usepackage{showlabels}
\usepackage{lineno}
\linenumbers
\fi

\IfFileExists{MinionPro.sty}
{\usepackage[lf]{MinionPro}
\def\interior{\righthalfcup}
\usepackage{amsmath,amsthm,amsbsy}}
{\usepackage{times}
\usepackage{savesym} 
\usepackage{amsmath,amsthm,amsbsy}
\savesymbol{iint}
\savesymbol{openbox}
\usepackage{txfonts}
\def\interior{\mathbin{\lrcorner}}
\DeclareSymbolFont{symbols}{OMS}{cmsy}{m}{n}
\restoresymbol{TXF}{iint}
\restoresymbol{TXF}{openbox}}

\usepackage{enumerate}
\usepackage{mathtools}
\usepackage{geometry}
\usepackage{color}
\definecolor{labelkey}{rgb}{0,0,.75}
\definecolor{MyGreen}{rgb}{0,.6,.2}
\definecolor{MyDarkBlue}{rgb}{.1,.1,.75}
\usepackage[colorlinks=true, citecolor=MyGreen, linkcolor=MyDarkBlue]{hyperref}
\usepackage{tensor}
\usepackage[all,cmtip]{xy}
\usepackage{graphicx}
\usepackage{mathrsfs}
\usepackage{mathtools}

\DeclareMathAlphabet\mathbfcal{OMS}{cmsy}{b}{n}

\DeclareMathOperator{\ck}{\mathbb{L}}
\renewcommand{\div}{\mathop{\rm div}\nolimits}

\DeclareMathOperator{\tr}{\rm tr}

\DeclareMathOperator{\id}{\rm id}

\def\ip<#1,#2>{\left<#1,#2\right>}

\newcommand{\Reals}{\mathbb{R}}
\newcommand{\Nats}{\mathbb{N}}
\newcommand{\Ints}{\mathbb{Z}}

\def\dc<#1,#2,#3>{\{#1;\;#2,#3\}}

\usepackage{etoolbox}

\title{Initial Data for First-order Causal Viscous Conformal Fluids in General Relativity}

\date{\today}

\author{Marcelo Disconzi\thanks{marcelo.disconzi@vanderbilt.edu} \\ Vanderbilt University \and James Isenberg\thanks{isenberg@uoregon.edu} \\ University of Oregon \and David Maxwell\thanks{damaxwell@alaska.edu} \\ University of Alaska Fairbanks}

\begin{document}
\newtheorem{theorem}{Theorem}[section]
\newtheorem{meta-theorem}{Meta-Theorem}[section]
\newtheorem{conjecture}[theorem]{Conjecture}
\newtheorem{problem}[theorem]{Problem}
\newtheorem{proposition}[theorem]{Proposition}
\newtheorem{corollary}[theorem]{Corollary}
\newtheorem{remark}[theorem]{Remark}
\newtheorem{lemma}[theorem]{Lemma}
\theoremstyle{definition}
\newtheorem{definition}[theorem]{Definition}
\numberwithin{equation}{section}

\maketitle

\begin{abstract}
We solve the Einstein constraint equations for a first-order causal viscous relativistic hydrodynamic theory in the case of a conformal fluid. For such a theory, a direct application of the conformal method does not lead to a decoupling of the equations, even for constant-mean curvature initial data. We combine the conformal method applied to a background perfect fluid theory with a perturbative argument in order to obtain the result.
\end{abstract}

\section{Introduction}
This paper is concerned with construction of general-relativistic initial-data sets for the first-order causal conformal viscous hydrodynamics theory introduced in \cite{Bemfica-Disconzi-Noronha-2018}, which corresponds to a relativistic version of the Navier-Stokes equations in the case of a conformal, pure-radiation, fluid. This theory has become known in physics literature as the BDNK theory \cite{Hoult:2020eho}. Such a theory is defined by a stress-energy tensor involving up to first-order derivatives of the fluid variables, see equation \eqref{E:Stress_energy_def} below, and is thus referred to as a first-order theory\footnote{Notice that the order refers to the order of the stress-energy tensor, not of the equations of motion.}. Its local well-posedness, causality, and linear stability have been proven in \cite{Bemfica-Disconzi-Noronha-2018,Disconzi-2019,Bemfica-Disconzi-Rodriguez-Shao-2021}, both in a fixed background and upon coupling to Einstein's equations, assuming in the latter case that an initial-data set exists, i.e., assuming that one can construct initial data for the Einstein-fluid system satisfying the constraint equations. Here, we establish the existence of such initial-data in both the compact and asymptotically Euclidean cases.

The viscous fluid theory we study is defined by the stress-energy tensor
\begin{align}
T^{ab} = \left(\epsilon + B\right)\left(u^a u^b + \frac{1}{3}\Delta^{ab}\right) 
- 2 \eta \sigma^{ab} + \mathcal Q^a u^b + \mathcal Q^b u^a,
\label{E:Stress_energy_def}
\end{align}
where $\epsilon$ is the fluid's equilibrium energy density, $u^a$ is a future-pointing timelike unit vector field representing the fluid's (four-)velocity, and $\Delta^{ab}$ is the the projection onto the space orthogonal to $u^a$ (see equation \eqref{eq:def-delta}). The equation of state is that of pure radiation, namely, $p= \epsilon/3$, where $p$ is the fluid's pressure which has already been substituted for $\epsilon$ in equation \eqref{E:Stress_energy_def}. The quantities  $B$, $\eta$, $Q^b$, and $\sigma^{ab}$  are, respectively, the viscous correction to the (equilibrium) energy density, the coefficient of shear viscosity, the heat flux, and the shear stress tensor (so that the combination $-2\eta \sigma^{ab}$ corresponds to the viscous shear stress). They model viscous corrections to the perfect fluid stress-energy tensor which is recovered upon setting $B$, $\eta$, and $Q^a$ to zero. The precise form of  $B$, $\eta$, $Q^b$, and $\sigma^{ab}$ is given in Section \ref{S:Viscous_conformal_model}. Here, it suffices to point out that they are chosen so that the model \eqref{E:Stress_energy_def} defines a conformal theory, i.e., a theory 
with a trace-free stress-energy tensor
such that upon conformal changes of the metric, $\tilde g_{ab} = e^{-2 \Omega} g_{ab}$ and an appropriate transformation of the fluid variables by powers of $e^{\Omega}$ (see Section \ref{secsec:scaling-viscous}), the divergence of $T^{ab}$ transforms as
\begin{align}
\label{E:Conformal_transformation_div_T}
\widetilde{\nabla}_a \widetilde{T}^{a}_b =e^{4\Omega}\nabla_a T^{a}_b,
\end{align}
so that solutions are preserved under conformal changes.

The need for the construction of a theory of relativistic fluids that accounts for viscous and dissipative\footnote{It is common practice in the community to use the terms viscous and dissipative interchangeably, save for specific cases where a clear distinction is warranted. We will thus speak only of viscosity when referring to what would be more accurately described as viscosity and dissipation.} processes has been known for quite a long time, since out-of-equilibrium effects can occur in relativistic regimes. While theories of relativistic viscous fluids go back to Eckart \cite{Eckart:1940te} and Landau and Lifshitz \cite{Landau-Lifshitz-Book-Fluids-1987} in the '40s and '50s, it was the discovery of the quark-gluon plasma \cite{Adams:2005dq,Adcox:2004mh,Back:2004je,Arsene:2004fa,Gyulassy:2004zy} that has brought the study of relativistic viscous hydrodynamics into the mainstream. The quark-gluon plasma is an exotic state of matter that forms in collisions of heavy ions, and a combination of experimental results, numerical simulations, and theoretical investigations, established that it behaves as a relativistic liquid with viscosity \cite{Romatschke:2017ejr}.
While the quark-gluon plasma can be studied within special relativity, recent studies \cite{Alford:2017rxf,Chabanov:2023blf,Most:2022yhe} strongly suggest that viscous effects are also relevant in mergers of neutron star, thus requiring a fully general-relativistic theory of fluids with viscosity.

Incorporating viscous effects into relativistic theories, however, has proven to be a daunting task. Several relativistic viscous fluid models, including the aforementioned Eckart and Landau and Lifshitz models, have been proven to be acausal \cite{Pichon-1965,Hiscock:1985zz}, thus incompatible with the very foundations of relativity. These models have also been proven to be linearly unstable \cite{Hiscock:1985zz}, where linear stability refers to mode stability of constant states. This is at odds with a basic physical intuition that dissipation due to viscosity should have a damping effect. Over the years, many different theories have been proposed to address these problems. A discussion of these different attempts, as well as the reasons why it seems difficult to construct viscous theories compatible with relativity, is beyond the scope of this article. We refer the interested reader to \cite[Chapter 6]{Rezzolla-Zanotti-Book-2013}, \cite[Section II]{Bemfica:2020zjp}, \cite{Baier:2007ix}, \cite{Romatschke:2017ejr}, \cite[Section 6]{Disconzi-2023-arxiv}, \cite{Gavassino:2021owo} and references therein.

The theory defined by the stress-energy tensor \eqref{E:Stress_energy_def} 
is important, therefore, since the previously mentioned 
local well-posedness and related results 
of \cite{Bemfica-Disconzi-Noronha-2018,Disconzi-2019,Bemfica-Disconzi-Rodriguez-Shao-2021}
show that it is a direct generalization of the Navier-Stokes equations
(in the case of a conformal fluid)
that is fully compatible with the principles of relativity.  It has 
consequently attracted much recent attention. Numerical simulations of its 
special relativistic dynamics have been carried out in \cite{Pandya:2021ief,Pandya:2022pif,Bantilan:2022ech,Bea:2023rru}. In Minkowski background, the corresponding equations of motion have also been showed to be small-data globally well-posed \cite{Sroczinski-2024}. The results \cite{Sroczinski-2024}, in fact, are more general, being applicable also to the theories introduced in \cite{Freistuhler-Temple-2014,Freistuhler-Temple-2017,Freistuhler-Temple-2018,Freistuhler-2020}, which are especially suited for the studies of shocks in relativistic viscous fluids \cite{Freistuhler-2021,Freistuhler-Temple-2023-arxiv,Pellhammer-2023}.
See also the related results \cite{Freistuhler-2023-arxiv,Sroczinski-2019,Sroczinski-2020,Freistuhler-Sroczinski-2021,Freistuhler-Reintjes-Sroczinski-2022,Freistuhler-Sroczinski-2024-arxiv}.

Although we focus only on the case of a conformal fluid,
we remark that the model \eqref{E:Stress_energy_def} has been extended to 
general, non-conformal fluids and arbitrary equations of state in \cite{Bemfica:2020zjp,Bemfica-Disconzi-Noronha-2019-1,Kovtun:2019hdm,Hoult:2020eho} (see also the related works \cite{Hoult:2021gnb,Rocha:2022ind,Taghinavaz:2020axp}). These works establish causality and linear stability of the general theory, with local well-posedness obtained in \cite{Bemfica:2020zjp} based on techniques developed in \cite{Bemfica-Disconzi-Graber-2021} (see \cite{Disconzi-Shao-2023-arxiv} for a simpler proof). Consequently, these extended models
provide a more comprehensive generalization of the Navier-Stokes equations to relativity than the formulation studied here.
For the Einstein-fluid 
systems arising from these models, local well-posedness requires initial data 
satisfying the Einstein constraint equations and hence there remains a 
need for initial data in this broader setting.

Nevertheless, it is reasonable to focus on the conformal case as a
proof of concept for how to address the general setting.  It is the simplest realization of the general theory developed in \cite{Bemfica:2020zjp,Bemfica-Disconzi-Noronha-2019-1,Kovtun:2019hdm,Hoult:2020eho}, and
despite this simplicity it has physical application to the quark-gluon plasma
discussed above. Moreover,
because the theory defined by model \eqref{E:Stress_energy_def} is conformally invariant in spacetime, and because our initial data construction applies the conformal method,
there is the possibility of leveraging conformal structures that might not be present
in the general setting. This indeed turns out to be the case, and our analysis makes
heavy use of the Weyl derivative, a conformally covariant connection on spacetime 
discussed in Section \ref{sec:weyl}.  In fact, the relationship between the conformal method (which involves conformal transformations of objects defined on
initial data sets) and conformal transformations of spacetime itself\kern 0.1em \footnote{We use the single word `conformal' in these two different senses throughout the paper. Both meanings are standard in the literature and they are readily  distinguished in context.} is not fully
understood, and working with a non-trivial conformally invariant matter model
provides an opportunity to shed light on this relationship.  As discussed below,
we find that there is a single key sticking point related to the two different conformal approaches to the trace-free component of the second fundamental form that restricts the scope of our results to the perturbative, small viscosity setting.

An initial-data set for the Einstein-matter equations consists of 
$(h_{ab},K_{ab},\mathcal{F})$ specified on a three dimensional manifold $\Sigma$, where $h_{ab}$ is a Riemannian metric, $K_{ab}$ is a symmetric 2-tensor representing the second fundamental form of $\Sigma$ embedded in an ambient spacetime, and $\mathcal{F}$ represents the various non-gravitational matter fields. If the initial-data set is sufficiently regular and satisfies the Einstein-matter constraint equations
\begin{align}
R_{h} - |K|_{h}^2 + (\tr_{h} K)^2 &= 2\mathcal{E}
&\text{\small [Hamiltonian constraint]}
    \label{eq:ham-constraint}\\
-\nabla^a K_{ab} &= {\mathcal J}_a
&\text{\small [momentum constraint]}
\label{eq:mom-constraint}
\end{align}
then it follows from the work of Choquet-Bruhat \cite{FouresBruhat:1952wg} and from subsequent work
(see, for example \cite{MR3287025} and \cite{Choquet-Bruhat-Book-2009})
that for a wide collection of Einstein-matter field theories (including the Einstein-Maxwell, Einstein-Klein-Gordon and Einstein-perfect fluid theories), there exists a spacetime solution of the corresponding Einstein-matter field equations which is consistent with the initial-data set. Above, 
$R_h$ is the scalar curvature of $h_{ab}$ and
$\mathcal{E}$ and $\mathcal{J}_a$ are the energy and momentum densities determined by $\mathcal{F}$ and $h_{ab}$; see Section \ref{S:Energy_momentum_density} for the specific form of $\mathcal{E}$ and $\mathcal {J}_a$ for conformal fluids.

The conformal method initiated by Lichnerowicz \cite{Lichnerowicz:1944th} and further developed by
Choquet-Bruhat \cite{MR0129941}
and York \cite{york_conformally_1973}
is the most successful technique for constructing and parameterizing initial-data sets satisfying equations \eqref{eq:ham-constraint}--\eqref{eq:mom-constraint}.  
It is especially effective when generating constant mean curvature
(CMC) data sets \cite{isenberg_constant_1995}, those for which the trace of $K_{ab}$ is constant,
and has been successfully been applied to numerous matter models. 
See, e.g., \cite{isenberg_gluing_2005} for constructions involving gravity coupled to Maxwell, Yang-Mills, and Vlasov fields.  In all these cases, the key step is to find 
a conformal scaling of the matter fields such that, as discussed in Section \ref{S:Scaling_matter}, for CMC initial data 
the momentum constraint decouples from the Hamiltonian constraint as a linear PDE which can be solved independently.  

Building on earlier physical intuition from \cite{isenberg_extension_1977} and \cite{kuchar_kinematics_1976}, 
the recent work \cite{IsenbergMaxwellMatter} rigorously demonstrates how to ensure
decoupling of the momentum constraint for matter models arising from a diffeomorphism invariant matter Lagrangian weakly coupled to gravity. 
Although the conformal method has previously been applied to 
Einstein-perfect fluid initial data sets \cite{dain_initial_2002a}\cite{isenberg_gluing_2005}, the scaling
laws introduced in those earlier works are ad-hoc and motivated solely by 
the pragmatic aim of obtaining decoupling.  It turns out that 
these ad-hoc rules do not compose
well with the scaling laws needed for other matter fields and cannot be used,
for example, to generate Einstein-Maxwell-charged dust initial data sets where the 
dust has a prescribed charge-to-mass ratio. By contrast, \cite{IsenbergMaxwellMatter}
develops a completely different set of scaling rules for perfect fluids, each 
with a character depending strongly on the fluid's constitutive relation, 
and having the advantage that they do compose well with other matter models.  The specific cases of dust and stiff fluids are treated in 
some detail in \cite{IsenbergMaxwellMatter}, but an analysis for general 
constitutive relations is nontrivial and only now in progress \cite{AllenMaxwell}.
In Section \ref{secsec:scaling-perfect} we show how to apply the 
\cite{IsenbergMaxwellMatter} construction to obtain scaling laws for 
perfect conformal fluids. 

The existence of a characterizing Lagrangian is not known to be true for viscous fluids (see above references) and moreover, if one were found, it would violate
the weak coupling hypothesis of \cite{IsenbergMaxwellMatter}.  Hence we rely on 
physical intuition to extend the perfect fluid scaling relations of Section \ref{secsec:scaling-perfect} to the viscous case in Section \ref{secsec:scaling-viscous}.
Guided by spacetime conformal transformations that lead to equation \eqref{E:Conformal_transformation_div_T}, we find conformal transformations 
that nearly decouple the momentum constraint with a single term 
involving the trace-free component $\theta_{ab}$ of the second fundamental
remaining as an obstacle. As discussed in Section \ref{secsec:scaling-viscous},
this term scales favorably when viewed from the spacetime perspective, but 
its very different role in the conformal method prevents the desired
decoupling.  Notably, $\theta_{ab}$ arises in our equations via the shear tensor 
$\sigma_{ab}$, which is a standard component of viscous fluid models, 
and hence the obstacle encountered here seems fundamental.

In order to overcome these difficulties, we combine the conformal method 
with a perturbative argument based on the Implicit Function Theorem to solve the constraint equations for the viscous fluid model. We first use the fact that the 
perfect conformal fluid model has excellent solvability properties when using 
the conformal method to construct baseline inviscid solutions of the constraints,
and then perturb the viscosity parameters to obtain nearby viscous solutions. 
While limited in scope, our results are consistent with the underlying physical motivation of model \eqref{E:Stress_energy_def} in that it describes first-order viscous corrections to a perfect fluid.  Once an existence theory for the
conformal method applied to perfect fluids with more general constitutive relations
becomes available, we expect these same techniques will apply in the general setting.
We only treat the case of CMC initial data, in part to highlight the specific 
new difficulties coming from the matter model.  The interested reader could 
modify our arguments with an additional Implicit Function Theorem step as in \cite{MR1179734} to generate
near-CMC initial data as well.

This paper is organized as follows. 
A key role in our discussion of the Einstein-conformal viscous fluid field theory is played by the Weyl covariant derivative. This derivative operator is conformally covariant with conformal transformations of tensor fields, including the spacetime metric. We define the Weyl derivative in Section \ref{sec:weyl}, and derive some of its properties in that section. In Section \ref{S:Conformal_model}, we introduce conformal fluid field theories, both inviscid and those with viscosity. In addition, in that section, we discuss the representation of the spacetime fluid field variables in terms of field variables defined spatially, and derive expressions 
for $\mathcal{E}$ and $\mathcal{J}_a$ in terms of these spatial fluid field variables. Then in Section \ref{S:Scaling_matter} we describe
the conformal method and apply conformal transformations to the spatial fluid field variables, both for the inviscid and the viscous cases. 
The primary results of our work appear in Section \ref{S:Initial_data}
where we establish the existence theory for the Einstein-perfect conformal fluid constraint equations along with the application of the Implicit Function Theorem
to construct solutions of the Einstein-conformal viscous fluid constraint equations for sufficiently small values of the viscosity parameters. These constructions
are performed in both the compact and asymptotically Euclidean settings.
Section \ref{S:Initial_data} concludes with a brief discussion of how 
the coordinate-free initial data constructed here adapts to the 
coordinate-based approach of the evolution results of, e.g., \cite{Bemfica:2020zjp}.

\subsection*{Acknowledgments} MD acknowledges support from NSF grants DMS-2406870 and DMS-2107701, from DOE grant DE-SC0024711, from a Chancellor's Faculty Fellowship, and from a Vanderbilt's Seeding Success Grant. JI acknowledges support from NSF grant PHY-1707427.

\section{Notation and preliminaries}

We assume spacetime is a connected 4-manifold $M$ and write $g_{ab}$ for Lorenztian 
metrics on $M$ with signature $(-,+,+,+$). Note that we are using abstract index notation \cite{Wald:1984rg}
here and elsewhere, except in contexts where it clutters the notation,
e.g. for the volume element $dV_g$ or for certain norms $|T|_g$.  As needed we assume that $M$ is
time-oriented and oriented.  

Spacelike hypersurfaces of $M$ are denoted by $\Sigma$.  These are connected,
oriented 3-manifolds, and many of the constants appearing in our equations
(e.g. the LCBY equations \eqref{eq:LCBY-ham-matter}--\eqref{eq:LCBY-mom-matter}
of the conformal method) are specific to the dimension 3.  Riemannian 
metrics on $\Sigma$ are $h_{ab} = g_{ab}+N_a N_b$, where $N^a$ is the future-pointing normal to $\Sigma$.  Note that $N^a h_{ab}=0$ and we
have arranged so that all tensors intrinsic to $\Sigma$ annihilate
the normal direction when extended to $M$.  In particular, 
at times we require an extension of the exterior derivative
on $\Sigma$ to $M$, and we write it as $d^\Sigma$ to 
emphasize our convention that $N^a(d^\Sigma f)_a=0$.

The second fundamental form $K_{ab}$ of $\Sigma$ is defined 
so that the Lie derivative of $h_{ab}$ in the normal direction is $+2K_{ab}$.
As a consequence, if $X^a$ and $Y^a$ are vector fields on $\Sigma$, then
$X^a \nabla^{(g)}_a Y^b = X^a \nabla^{(h)}_a Y^b + X^a Y^c K_{ac} N^b$.
We drop the notation $(g)$ and $(h)$ on the covariant derivative when it is clear
from context which is meant.

The metric $h_{ab}$ induces the scalar curvature $R_h$, the Laplacian 
$\Delta = h^{ab}\nabla_a \nabla_b$ and the conformal Killing operator 
$\ck$ defined by
\[
(\ck W)_{ab} = \nabla_a W_b + \nabla_b W_a - \frac{2}{3} \nabla_b W^b h_{ab}.
\]

We frequently work with metrics $\tilde g_{ab}$ on $M$ that are conformally related
to $g_{ab}$ via $\tilde g_{ab} = e^{-2\Omega} g_{ab}$ for some function $\Omega$;
the same notation is used for metrics $h_{ab}$ on $\Sigma$.
Operators $\Delta$ and $\ck$ are decorated with tildes if they are associated
with $\tilde h_{ab}$, and index raising or lowering of tensors decorated with 
tildes is performed with metrics with tildes.

\section{Weyl derivative and second fundamental form}\label{sec:weyl}

The stress-energy tensor for the viscous conformal fluid model
is efficiently expressed in terms of a Weyl derivative,
a conformally
covariant differential operator that we briefly review here.  
The general construction is due to \cite{weyl_gravitation_1923} in the context of conformal field theory; 
see also \cite{folland_weyl_1970}\cite{hall_weyl_1992}.
Its specific application to conformal hydrodynamics 
appears in \cite{loganayagam_entropy_2008}. Thus, before presenting the definition of the quantities $B$, $\eta$, $Q^b$, and $\sigma^{ab}$ in Section \ref{S:Viscous_conformal_model}, we introduce the Weyl derivative now and establish some properties that will be proven useful later.

Consider a conformal class $\mathbf g$ of Lorentzian metrics
on $M$.  A conformally weighted tensor field $\mathbf T^{a\cdots}_{b\cdots}$ 
is a map on $\mathbf g$ taking a representative $g_{ab}$ to a tensor $T^{a\cdots}_{b\cdots}$
such that $\tilde g_{ab} = e^{-2 \Omega} g_{ab}$ is taken to $e^{w\Omega}T^{a\cdots}_{b\cdots}$
for some number $w$ called the weight of the tensor field.  For example,
the conformal metric itself has weight $-2$ and its 
spacetime volume form $dV_g$ has weight $-4$ in $4$ spacetime dimensions.
As another example, suppose $X^a$ is a timelike vector field on $M$.
The unit vector field $u^a$ parallel to $X^a$ has weight one since the condition
$\tilde g_{ab}\tilde u^a \tilde u^b = -1$ requires $\tilde u^a = e^{\Omega} u^a$.
Although in each of these examples 
the weight happens to be the number of contravariant indices minus
the number of covariant indices, we encounter below other possibilities 
such as scalar fields with a non-zero weight.

In the context of hydrodynamics a Weyl derivative 
is a differential operator $\mathbfcal{D}_a$ 
determined\footnote{More invariantly, the Weyl derivative is 
determined by $\mathbf g$ and a foliation of timelike curves; the vector
field $X^a$ can be taken to be any future pointing vector field which is tangent to the
foliation.} by a conformal class $\mathbf g$ and a 
time-like vector field $X^a$, and it acts on conformally weighted tensor fields.
To describe this action, we select a representative $g_{ab}$ of $\mathbf g$
and write down a differential operator $\mathcal D_a$ that represents
$\mathbfcal D_a$ with respect to $g_{ab}$ as follows.

Let $u^a$ be the future-pointing unit vector field 
parallel to $X^a$ and define
\begin{equation}\label{eq:A-def}
    \mathcal A_a = u^b \nabla_b u_a - \frac{1}{3} \nabla_b u^b u_a
\end{equation}
using the Levi-Civita connection $\nabla_a$ of $g_{ab}$.
In applications, $u^a$ is the velocity field of a fluid and $\mathcal A_a$
is a measure of its acceleration. Now let $\mathbf T^{a\cdots}_{b\cdots}$ be a conformally weighted tensor field
with weight $w$ and let $T^{a\cdots}_{b\cdots}$ be its representative with
respect to $g_{ab}$.  Then
\begin{equation}\label{eq:weyl-def}
\mathcal D_c T^{a\cdots}_{b\cdots} = \nabla_c T^{a\cdots }_{b\cdots} + w \mathcal A_c T^{a\cdots}_{b\cdots}
+C^a_{cd} T^{d\cdots}_{b\cdots}+\cdots-
C^d_{cb} T^{a\cdots}_{d\cdots}-\cdots
\end{equation}
where
\[
C_{bc}^a = g_{cb}g^{ad}\mathcal A_d - \delta_c^a \mathcal A_b - \delta^a_b \mathcal A_c.
\]
One readily verifies
\begin{equation}\label{eq:A-conf-change}
    \widetilde{\mathcal  A}_a = \mathcal A_a - (d\Omega)_a
\end{equation} 
under the conformal transformation $\tilde g_{ab} = e^{-2\Omega} g_{ab}$
and that consequently
\begin{equation}\label{eq:Weyl-conf}
\widetilde {\mathcal D}_c \widetilde T^{a\cdots}_{b\cdots} = e^{w\Omega} 
{\mathcal D}_c T^{a\cdots}_{b\cdots}.
\end{equation}
That is, the Weyl derivative determines a conformally weighted tensor 
$\mathbfcal{D}_c \mathbf T^{a\cdots}_{b\cdots}$ with the same weight
as $\mathbf T^{a\cdots}_{b\cdots}$.

The Weyl derivative is compatible with the conformal class in the sense
that $\mathcal D_c g_{ab} = 0$.  Although it is not true in general that 
$\mathcal D_c u^a$ vanishes, this quantity nevertheless enjoys
a number of useful properties. To describe these, we introduce the projection
\begin{equation}\label{eq:def-delta}
\Delta_a^b = \delta_a^b + u_au^b    
\end{equation}
onto the subspace orthogonal to $u^a$, so $u^a \Delta_a^b=0$ and
$u_b \Delta_a^b=0$. Direct computation shows
\begin{equation}\label{eq:D-u}
\mathcal D_a u^b = \Delta_a^c \nabla_c u^b - \frac{1}{3}\nabla_c u^c \Delta_a^b
\end{equation}
from which we obtain the useful identities
\begin{align}
    u^a \mathcal D_a u^b &= 0\\
    u_b \mathcal D_a u^b &= 0\\
    \mathcal D_a u^a &= 0;\label{eq:shear-trace-free}
\end{align}
and deriving this last equation requires the observation $(\nabla_a u^b)u_b=0$.

The shear tensor of $u^a$ defined by
\begin{equation}\label{eq:shear-v1}
\sigma_{ab} = \frac12 \left( \Delta_a^c \nabla_c u_b + \Delta_b^c \nabla_c u_a  \right)
- \frac{1}{3}\nabla_c u^c \Delta_{ab}
\end{equation}
plays an important role in viscous fluid models generally. Comparing
equation \eqref{eq:shear-v1} with equation \eqref{eq:D-u} we find
that the shear is easily written in terms of the Weyl derivative:
\begin{equation}\label{eq:shear-v2}
\sigma_{ab} = \frac{1}{2}\left( \mathcal D_a u_b + \mathcal D_b u_a\right).
\end{equation}

The vector field $u^a$ defining the Weyl derivative 
is implicit in the notation and we introduce the 
following additional (implicit) notation for derivatives
parallel to and orthogonal to $u^a$:
\begin{equation}\label{eq:D-angle-D}
\begin{aligned}
\mathcal D &= u^a \mathcal D_a\\
\mathcal D_{\left<a\right>} &= \Delta^b_a \mathcal D_b
\end{aligned}
\end{equation}
where $\Delta^b_a$ is the projection introduced in equation \eqref{eq:def-delta}.

Although the special case of the Weyl derivative given above is well-adapted for
use in hydrodynamics, the Weyl derivative admits a more general definition in terms of an arbitrary
one form $\mathcal A_a$ that conformally transforms according to the rule
\eqref{eq:A-conf-change}.  Equation \eqref{eq:weyl-def} applies equally in
this case to define the Weyl derivative, which again determines a well-defined
map on conformally weighted tensor fields and for which $\mathbfcal D_a \mathbf g = 0$.
A (hydrodynamic) Weyl derivative on $M$ induces a (general) Weyl derivative 
on hypersurfaces of $M$, which we examine next.

Let $\Sigma$ be a spacelike hypersurface of $M$ with future-pointing timelike normal $N^a$.  
The metric $g_{ab}$
and the 1-form $\mathcal A_a$ induce a metric $h_{ab}$ and a 1-form $\mathcal A_a^\Sigma$
on $\Sigma$ by pulling back via the natural embedding, and we extend 
these tensors in non-tangential directions by declaring that they annihilate $N^a$.
The conformal transformation 
$g_{ab}\mapsto \tilde g_{ab} = e^{-2\Omega}g_{ab}$ leads to 
$\tilde h_{ab} = e^{-2\Omega}h_{ab}$ and 
$\widetilde {\mathcal A}_a^{\,\Sigma} = \mathcal A_a^\Sigma - (d^\Sigma \Omega)_a$
where $(d^\Sigma\Omega)_a := h^b_a (d\Omega)_b$.
Hence, when restricted to $\Sigma$, $\mathcal A_a^\Sigma$ has the conformal transformation rule
\eqref{eq:Weyl-conf} and
$\mathbfcal D_a$ restricts to a generalized Weyl derivative $\mathbfcal D^\Sigma_a$
on $\Sigma$.

We define the Weyl second fundamental form $\mathcal K_{ab}$ 
analogously to the usual metric construction.
If $X^a$ and $Y^b$ are vector fields tangent to $\Sigma$ we set
\[
\mathcal K_{ab} X^a Y^b = - g_{cb} N^c X^a \mathcal D_a Y^b
\]
and we extend it by declaring that it annihilates $N^a$.
\begin{lemma}\label{lem:weyl-K} The Weyl second fundamental form satisfies
\begin{equation}\label{eq:weyl-K}
\mathcal K_{ab} = K_{ab} - h_{ab} \mathcal A_c N^c.
\end{equation}
where $K_{ab}$ is the second fundamental form of $\,\Sigma$. 
In particular it is a symmetric tensor field of weight $-1$.
Moreover, for a vector $X^a$ on $\Sigma$ and a vector field $Y^b$
on $\Sigma$ of any weight,
\begin{equation}\label{eq:K-split-normal}
X^a \mathcal D_a Y^b = X^a \mathcal D_a^\Sigma Y^b + \mathcal{K}_{ac}X^aY^c N^b
\end{equation}
and
\begin{equation}\label{eq:deriv-N}
 X^a g_{bc} \mathcal D_a N^b Y^c = \mathcal K_{ab} X^a Y^b.
\end{equation}
\end{lemma}
\begin{proof}
To establish equation \eqref{eq:weyl-K} it suffices to show that it holds for vectors
tangential to $\Sigma$ since both sides vanish when applied in the normal direction.
From equation \eqref{eq:weyl-def}, if $X^a$ and $Y^b$ are vector fields along $\Sigma$,
and if $Y^b$ has weight $w$,
\begin{equation}
X^a \mathcal D_a Y^b = X^a \nabla_a Y^b + w X^a \mathcal A_a Y^b 
+ g_{ac} X^a Y^c g^{bd}\mathcal A_d - X^b \mathcal A_a Y^a - Y^b \mathcal A_a X^a.
\end{equation}
From the above equation and the relations
\begin{align*}
    X^a \nabla_a Y^b &= X^a \nabla^\Sigma_a Y^b + K_{ac} X^a Y^c N^b\\
    g^{bd}\mathcal A_d &= -N^b \mathcal A_c N^c + h^{bd}\mathcal A^\Sigma_d    
\end{align*}
along with the trivial identities $\mathcal{A}_a X^a = \mathcal{A}^\Sigma_a X^a$,
$\mathcal{A}_a Y^a = \mathcal{A}^\Sigma_a Y^a$ and $g_{ab} X^a Y^b = h_{ab}X^a Y^b$ 
we find
\begin{align*}
X^a \mathcal D_a Y^b &= \left[  
X^a \nabla^\Sigma_a Y^b + w X^a \mathcal A_a^\Sigma Y^b
+h_{ac} X^a Y^c h^{bd}\mathcal A^\Sigma_d - X^b \mathcal A^\Sigma_a Y^a - Y^b \mathcal A_a^\Sigma X^a
\right]
 + \left(K_{ad} - h_{ad} \mathcal A_c N^c\right) X^a Y^d N^b\\
 &=  X^a \mathcal D^\Sigma_a Y^b + (K_{ad} - h_{ad} \mathcal A_c N^c) X^a Y^d N^b.
\end{align*}
Equation \eqref{eq:weyl-K} follows from taking a dot product with $N^b$
and equation \eqref{eq:K-split-normal} is an immediate consequence.
To demonstrate equation \eqref{eq:deriv-N} we note that since $Y_c N^c= 0$
and since $\mathcal D_a g_{bc} =0$, 
\[
0 = X^a \mathcal D_a (g_{bc} Y^b N^c) = X^a g_{bc} (\mathcal D_a Y^b) N^c
+ X^a g_{bc} Y^b (\mathcal D_a N^c).
\]
Equation \eqref{eq:deriv-N} follows from the above and the
definition of $\mathcal K_{ab}$.
\end{proof}

\section{Conformal fluid model}\label{S:Conformal_model}

Our method for constructing initial-data sets for the Einstein-viscous fluid model involves perturbing off data for the Einstein-perfect fluid model. Thus, we begin by discussing the perfect fluid case.

\subsection{Perfect conformal fluids}
The dynamical variables of a perfect fluid can be taken to be
a timelike unit vector $u^a$ 
and a scalar rest energy density
$\epsilon$.  In terms of these variables, the associated
stress-energy tensor is 
\[
T^{ab} = \epsilon u^a u^b + p (g^{ab}+u^au^b)
\]
where the pressure $p$ is determined
by a constitutive relation $p = p(\epsilon)$ that characterizes 
the fluid. The direction $u^a$ determines the fluid's rest frame in which 
an observer sees no momentum flux and is therefore co-moving with the fluid
energy. 

The fluid is imagined to be consisting of identical particles with
number density $n$ measured in the fluid rest fame. 
The rest energy density and rest number density
are not independent and, as discussed in
Section \ref{secsec:scaling-perfect} below, 
the constitutive relation for the fluid can 
alternatively be specified in the form $\epsilon = \epsilon(n)$,
in which case
\begin{equation}\label{eq:pressure-from-n}
    p(n) = n \epsilon'(n) - \epsilon(n).        
\end{equation}
Note that this is the first law of thermodynamics for
an isentropic fluid (see \cite{misner_gravitation_2017} equations (22.6) and (22.7a)).  
The number flux of the fluid is
$j^a = n u^a$ and we can hence interpret $u^a$
as the velocity of the fluid.

The fluid temperature $T$ is related to the pressure $p$ and 
number density $n$ via
\begin{equation}\label{eq:ideal-gas-law}
p = n T    
\end{equation}
in units where Boltzmann's constant is 1. By contrast, the rest entropy 
density of the fluid is dynamically uninteresting and we can
take it to be constant.

The stress-energy tensor associated with a conformally invariant
Lagrangian is trace free, $g_{ab} T^{ab}=0$, which can
be derived by varying the Lagrangian with respect to a conformal factor.  
As a consequence, a conformal perfect fluid satisfies $\epsilon=3p$,
which is its constitutive relation, and its stress-energy
tensor simplifies to
\[
T^{ab} = \epsilon \left(u^a u^b + \frac{1}{3}\Delta^{ab}\right).
\]

Equation \eqref{eq:pressure-from-n} 
implies $n\epsilon'(n) = \frac{4}{3}\epsilon(n)$ if $\epsilon=3p$
and hence
\[
\epsilon(n) = 3 k_\epsilon n^{\frac{4}3}
\]
for some constant $k_\epsilon$. Equations \eqref{eq:pressure-from-n}  and 
\eqref{eq:ideal-gas-law} then imply
\begin{align*}
p &= k_\epsilon n^\frac43\\
T &= k_\epsilon n^\frac13.
\end{align*}
As an important consequence, 
the rest energy density is related to the temperature via
$\epsilon = c_\epsilon T^4$ with $c_\epsilon = (3/k_\epsilon^3)$.

\subsection{Viscous conformal fluid}\label{S:Viscous_conformal_model}

We are now ready to discuss details of the stress-energy tensor \eqref{E:Stress_energy_def}, which we state here again for the reader's convenience:
\begin{align}
T^{ab} = \left(\epsilon + B\right)\left(u^a u^b + \frac{1}{3}\Delta^{ab}\right) 
- 2 \eta \sigma^{ab} + \mathcal Q^a u^b + \mathcal Q^b u^a,
\nonumber
\end{align}
where $\Delta^{ab}$ is the projection from equation \eqref{eq:def-delta} and 
$\sigma^{ab}$ is the shear tensor from equation \eqref{eq:shear-v2}.
We recall from the introduction that $B$, $\eta$, and $Q^b$, are, respectively, the viscous correction to the equilibrium energy density, the coefficient of shear viscosity, and the heat flux, and that the combination $-2\eta \sigma^{ab}$ corresponds to the viscous shear stress. As discussed in \cite{Bemfica-Disconzi-Noronha-2018}, physical considerations and the requirement that one can enforce the transformation law \eqref{E:Conformal_transformation_div_T} imply that
 $B$ and $Q^a$ are given by
\begin{align*}
B = 3\chi \frac{u^a \mathcal D_a T}{T},
\\
\mathcal Q^a = \lambda \frac{\Delta^{ac} \mathcal D_c T}{T},
\end{align*}
where $T$ is the equilibrium temperature, and $\chi$ and $\lambda$ are two further (i.e., in addition to $\eta$) viscosity coefficients; $\lambda$, $\eta$ and $\chi$ each have the form
$c T^3$ for some constant $c=c_{\chi/\eta/\lambda}$ specific to the coefficient.
Recalling that
$\epsilon=c_\epsilon T^4$ for a perfect conformal fluid 
we can write the stress-energy tensor in terms of $\epsilon$ as
\begin{equation}\label{eq:T-vcf}
    T^{ab} = \left(\epsilon + \frac{3\chi}{4\epsilon} D \epsilon\right) \left(u^a u^b + \frac13 \Delta^{ab}\right) -2\eta \sigma^{ab} 
    + \frac{\lambda}{4\epsilon}\left(u^a \mathcal D^{\left<b\right>} \epsilon +  u^b \mathcal D^{\left<a\right>}\epsilon \right)        
\end{equation}
where $\chi = c_\chi (\epsilon/c_\epsilon)^{3/4}$, $\eta = c_\eta (\epsilon/c_\epsilon)^{3/4}$ 
and $\lambda = c_\lambda (\epsilon/c_\epsilon)^{3/4}$
and where we recall the notation \eqref{eq:D-angle-D}
for $\mathcal D_{\left<a\right>}\epsilon$ and $\mathcal D \epsilon$.
More explicitly, in terms of temperature,
\begin{equation}\label{eq:T-vcf-temp}
    T^{ab} = \left(c_\epsilon T^4 + 3c_\chi T^2\, \mathcal D T\right) \left(u^a u^b + \frac13 \Delta^{ab}\right) -2c_\eta T^3 \sigma^{ab} 
    + c_\lambda T^2 \left(u^a \mathcal D^{\left<b\right>} T +  u^b \mathcal D^{\left<a\right>} T \right).   
\end{equation}
Although equations \eqref{eq:T-vcf} and \eqref{eq:T-vcf-temp} are physically equivalent descriptions in view of the relation $\epsilon=c_\epsilon T^4$, in our case it will frequently be more convenient to work with \eqref{eq:T-vcf-temp}, in which case the fluid's dynamical variables consist of  $T$ and $u^a$.  Nevertheless,
we freely transfer between $\epsilon$ and $T$ as desired.

\subsection{Energy and momentum density}\label{S:Energy_momentum_density}

The right-hand sides of the constraint equations require 
the energy and momentum densities 
\begin{align*}
\mathcal E &= N^a N^b T_{ab}\\
\mathcal J_b &= -N^a h_b^c T_{ac}
\end{align*}
and we wish to express these quantities in terms of data
associated with the slice $\Sigma$.

To begin, the spacetime metric $g_{ab}$
induces the Riemannian metric $h_{ab}$ and the second fundamental
form $K_{ab}$. The functions $\epsilon$ and $T$ restrict naturally
to $\Sigma$, but $T_{ab}$ also involves their first derivatives. 
The tangential derivatives can be computed from the 
restrictions to $\Sigma$, but we require transverse derivatives,
and we find it useful to encode these by working with 
$\mathcal D\epsilon = u^a\mathcal D_a \epsilon$ and $\mathcal DT = u^a \mathcal D_a T$.  
This leaves us with the the fluid velocity $u^a$ which can be decomposed
\begin{equation}
    u^a = \gamma N^a + v^a
\end{equation}
where $N^a$ is the normal to $\Sigma$, $v^a$ is tangential to $\Sigma$,
and where $\gamma^2 = 1+|v|_h^2$ ensures that $u^a$ has unit length.
Note that $w^a = v^a/\gamma$ is the velocity of the fluid seen by an observer moving
orthogonal to $\Sigma$ and that $\gamma$ can be rewritten 
$\gamma = 1/\sqrt{1-|w|_h^2}$, which is the usual 
relativistic expansion factor. We require transverse derivatives of 
$u^a$ as well, and these are captured by the 1-form $\mathcal A_a$
from equation \eqref{eq:A-def} associated 
with the Weyl derivative.  Indeed, $\mathcal A_a$ is a measure
of the fluid acceleration and it induces $\mathcal A^\Sigma_a$ 
on $\Sigma$ and therefore also a Weyl derivative $\mathcal D_a^\Sigma$
as discussed at the end of Section \ref{sec:weyl}.
In this section we show how the energy and momentum densities
for the conformal fluid
can be written compactly in terms of the slice intrinsic fields
$h_{ab}$, $v^a$ and $\epsilon$ along with the quantities
$K_{ab}$, $\mathcal A^\Sigma_a$ and $\mathcal D\epsilon$, which
serve as proxies for their time derivatives. See also Section 
\ref{S:Cauchy} where we explicitly show how to extract 
from $(h_{ab}, K_{ab}, \epsilon, \mathcal D\epsilon, v^a, \mathcal A^\Sigma_a)$
the values and time derivatives of the quantities needed for solving the 
associated Cauchy problem.

To compute $\mathcal E$ and $\mathcal J_a$ we require
suitable normal and tangential contractions 
with each of the tensors
\begin{equation}
    \label{eq:three-T-pieces}
    \left(u_a u_b +\frac{1}{3}\Delta_{ab}\right),\quad \sigma_{ab},\quad \text{and}\quad 
    \left(u_a \mathcal D_{\left<b\right>}\epsilon 
    + u_b \mathcal D_{\left<a\right>}\epsilon
    \right)
\end{equation}
appearing in $T_{ab}$ from equation \eqref{eq:T-vcf}.
Of these, the first and last are easiest and are handled by the following lemma.
\begin{lemma}\label{lem:bulk-heat}
    The following relations hold:
    \begin{align}
    N^aN^b\left(u_a u_b+\frac13\Delta_{ab}\right) &= 1+\frac43 |v|_h^2,\label{eq:E-bulk}\\
    -N^a h_b^c \left(u_a u_c+\frac13\Delta_{ac}\right) &= \frac43\gamma v_b,\label{eq:J-bulk}
    \end{align}
and
    \begin{align}
        N^a N^b( u_a \mathcal D_{\left<b\right>}\epsilon + u_b \mathcal D_{\left<a\right>}\epsilon) &= 2 v^a(  \mathcal D^\Sigma_a \epsilon + v_a \mathcal D\epsilon),\label{eq:E-heat}\\
       -N^a h_b^c( u_a \mathcal D_{\left<c\right>}\epsilon + u_c \mathcal D_{\left<a\right>}\epsilon) &= \gamma\left(\delta^c_b + \frac{v_b v^c}{\gamma^2}\right)\left(\mathcal D^\Sigma_c\epsilon + v_c \mathcal D\epsilon\right) \label{eq:J-heat}.
    \end{align}    
\end{lemma}
\begin{proof}
Since $u^a \Delta_{ab} = 0$ and since $u^a = \gamma N^a + v^a$,
\begin{equation}\label{eq:N-Delta}
    N^a \Delta_{ab} =     
     -\frac{v^a}{\gamma} \Delta_{ab} =
     -\frac{v^a}{\gamma}(g_{ab} + u_a u_b) = -\frac{1}{\gamma}(v_b +|v|_h^2 u_b).
\end{equation}
Consequently 
\begin{equation}\label{eq:NN-Delta}
    N^a N^b \Delta_{ab} = |v|_h^2
\end{equation}
and
\begin{equation}\label{eq:Nh-Delta}
    -N^a h^c_b \Delta_{ac} = \frac{1}{\gamma} h^c_b( v_c +|v|_h^2 u^b) 
    = \frac{1}{\gamma}( 1+|v|^2_h)v_b = \gamma v_b.
\end{equation}
Equations \eqref{eq:E-bulk} and \eqref{eq:J-bulk} follow from
equations \eqref{eq:NN-Delta} and \eqref{eq:Nh-Delta} along with the
observations $N^a N^b u_a u_b =  \gamma^2$ and $-N^a h_b^c u_a u_c = \gamma v_b$.

Turning to the heat flux terms (i.e., the third tensor in \eqref{eq:three-T-pieces}), equation \eqref{eq:N-Delta} implies
\begin{equation}\label{eq:N-in-De}
N^a \mathcal D_{\left<a\right>}\epsilon = 
N^a \Delta_a^b\mathcal D_{b}\epsilon = -\frac{1}{\gamma}(v^a + |v|_h^2 u^a) \mathcal D_a\epsilon 
= -\frac{1}{\gamma}(v^a \mathcal D_a^\Sigma\epsilon + |v|_h^2 \mathcal D\epsilon).
\end{equation}
Hence
\begin{equation}\label{eq:Nh-heat2}
    -N^a h^c_b \left(u_c \mathcal D_{\left<a\right>}\epsilon \right) = \frac{1}{\gamma} v_b
    (v^a \mathcal D_a^\Sigma\epsilon + |v|_h^2 \mathcal D\epsilon)
    = \frac{1}{\gamma} v_b v^a (\mathcal D_a^\Sigma\epsilon + v_a \mathcal D\epsilon).        
\end{equation}
On the other hand,
\[
    h_b^c\mathcal D_{\left<c\right>} \epsilon = h_b^c (\delta_c^d + u_c u^d)\mathcal D_{d}\epsilon =
    \mathcal D_b^\Sigma\epsilon + v_b \mathcal D\epsilon
\]
and therefore
\begin{equation}\label{eq:Nh-heat1}
    -N^a h_b^c \left(u_a\mathcal D_{\left<c\right>}\epsilon\right) = \gamma \mathcal (\mathcal D_b^\Sigma\epsilon + v_b \mathcal D\epsilon).
\end{equation}
Equations  \eqref{eq:Nh-heat2} and \eqref{eq:Nh-heat1} then imply
\[
    -N^a h_b^c (u_a\mathcal D_{\left<c\right>}\epsilon + u_c\mathcal D_{\left<a\right>}\epsilon)
    = \gamma\left( \delta_b^a + \frac{v^a v_b}{\gamma^2}\right)
    (D_a^\Sigma\epsilon + v_a \mathcal D\epsilon)
\]
which is equation \eqref{eq:J-heat}. Moreover,
\[
N^aN^b \left(u_a\mathcal D_{\left<b\right>}\epsilon + u_b\mathcal D_{\left<a\right>}\epsilon\right)
= 2 N^aN^b u_a\mathcal D_{\left<b\right>}\epsilon = -2\gamma N^b \mathcal D_{\left<b\right>}\epsilon
\]
so equation \eqref{eq:N-in-De} implies
\[
N^aN^b \left(u_a\mathcal D_{\left<b\right>}\epsilon + u_b\mathcal D_{\left<a\right>}\epsilon\right)=
2 (v^a \mathcal D^\Sigma_a \epsilon + |v|_h^2 \mathcal{D} \epsilon)
\]
which is equivalent to equation \eqref{eq:E-heat}.
\end{proof}

It remains to analyze contributions to the energy and momentum densities
arising from the shear tensor $\sigma_{ab}$.  
Since $u^a \sigma_{ab} = 0$ and since $N^a = (u^a-v^a)/\gamma$,
it follows that 
\begin{equation}\label{eq:v-for-N}
    N^a \sigma_{ab} = -(v^a/\gamma)\sigma_{ab}   
\end{equation}
and hence $\sigma_{ab}$ is completely determined by its restriction to $\Sigma$.
The following lemma is the first step to computing this restriction.
\begin{lemma}\label{lem:sigma-eval-v1} 
If $X^a$ and $Y^b$ are tangent to $\Sigma$ then
\begin{equation}\label{eq:sigma-eval-v1}
\sigma_{ab} X^a Y^b = \left[ \gamma \mathcal K_{ab}  + \frac12\left(\mathcal D_a^\Sigma v_b + \mathcal D_b^\Sigma v_a\right)\right] X^a Y^b.
\end{equation}
\end{lemma}
\begin{proof}
Starting from $u_b = \gamma N_b + v_b$ we have
\[
\mathcal D_a u_b = (\mathcal D_a \gamma) N_b + \gamma (\mathcal D_a N_b) 
+ \mathcal D_a v_b.
\]
The fact that $N^a$ is orthogonal to $\Sigma$
and equations \eqref{eq:deriv-N} and \eqref{eq:K-split-normal} 
from Lemma \ref{lem:weyl-K} imply
\[
X^a Y^b \mathcal D_a u_b =  \gamma \mathcal K_{ab} X^a Y^b + 
X^a Y^b \mathcal D_a^\Sigma v_b
\]
and equation \eqref{eq:sigma-eval-v1} follows from symmetrizing.
\end{proof}

Equation \eqref{eq:sigma-eval-v1} and the definition 
\eqref{eq:weyl-K} of the Weyl second fundamental form
show that $\sigma_{ab}$ can be computed from the 
slice intrinsic
data $h_{ab}$, $v^a$, $\epsilon$, $K_{ab}$, $\mathcal A^\Sigma_a$, 
and $\mathcal D\epsilon$ along with one additional term 
arising in the conformal second fundamental form: $\mathcal A_c N^c$. 
The following result shows how this key term can, in fact, also be computed in
terms of the slice-intrinsic data on $\Sigma$. For convenience we introduce
the following notation: 
\begin{equation}\label{eq:shear-etc}
\begin{aligned}
\tau = h^{ab} K_{ab}&:\quad\text{mean curvature},\\
\theta_{ab} = K_{ab} - \frac{\tau}{3} h_{ab}&:\quad\text{trace-free component of the second fundamental form},\\
\sigma^{\Sigma}_{ab} = \frac12\left(\mathcal D^\Sigma_a v_b + \mathcal D^\Sigma_b v_a\right)-\frac13 \mathcal D_c^\Sigma v^c h_{ab}&:\quad\text{induced shear tensor}.
\end{aligned}
\end{equation}
\begin{lemma}\label{lem:N-on-A}
    On $\Sigma$,
\begin{equation}\label{eq:N-in-A}
\mathcal A_c N^c = -\frac{1}{1 + 2\gamma} \left(\theta_{ab}v^a v^b + \frac{1}{\gamma}
\sigma_{ab}^\Sigma v^a v^b\right) 
+ \frac{1}{3} \left(\tau + \frac{1}{\gamma} \mathcal D_a^\Sigma v^a\right).
\end{equation}
\end{lemma}
\begin{proof}
Equation \eqref{eq:sigma-eval-v1} and the definition \eqref{eq:weyl-K}
of the Weyl second fundamental form imply
\begin{equation}\label{eq:h-in-sigma}
\gamma h^{ab}\sigma_{ab} = \gamma^2\tau - 3 \gamma^2 \mathcal A_c N^c +\gamma \mathcal D_c^\Sigma v^c.
\end{equation}
On the other hand, equation \eqref{eq:v-for-N} and Lemma \ref{lem:sigma-eval-v1}
imply
\begin{equation}\label{eq:N-in-sigma}
\begin{aligned}
-\gamma N^a N^b \sigma_{ab} &= -\frac{1}{\gamma} v^a v^b \sigma_{ab}\\
& = - K_{ab} v^a v^b + |v|_h^2 \mathcal A_c N^c - v^a v^b \frac1\gamma \mathcal D^\Sigma_a v_b\\
& = -\theta_{ab}v^a v^b - \frac{1}{3} |v|_h^2 \tau 
+|v|_h^2 \mathcal A_c N^c - \frac1\gamma \sigma^{\Sigma}_{ab} v^a v^b - \frac{1}{3\gamma}|v|^2 \mathcal D^\Sigma_a v^a.
\end{aligned}
\end{equation}
Adding equations \eqref{eq:h-in-sigma} and \eqref{eq:N-in-sigma}
and observing that $\gamma^2 - (1/3)|v|^2 = (1+2\gamma^2)/3$ we find
\begin{equation}\label{eq:g-in-sigma}
\gamma (-N^aN^b + h^{ab})\sigma_{ab} =
-\theta_{ab}v^a v^b 
+\frac{1+2\gamma^2}{3}\tau
-\frac1\gamma \sigma^\Sigma_{ab} v^a v^b  
+\frac{1+2\gamma^2}{3\gamma}\mathcal D_c^\Sigma v^c
-(1+2\gamma^2) \mathcal A_c N^c.
\end{equation}
But equation \eqref{eq:shear-trace-free} implies $g^{ab}\sigma_{ab}=0$
and hence $(-N^a N^b + h^{ab})\sigma_{ab} = 0$ as well.
That is, the right-hand side of equation \eqref{eq:g-in-sigma}
vanishes and we can therefore solve for $\mathcal A_c N^c$
to obtain equation \eqref{eq:N-in-A}.
\end{proof}
A straightforward computation from equations \eqref{eq:sigma-eval-v1}
and \eqref{eq:N-in-A} now leads to our final expression for the 
shear tensor in terms of intrinsic data on $\Sigma$.
\begin{lemma}
If $X^a$ and $Y^b$ are tangent to $\Sigma$ then
\begin{equation}\label{eq:shear-final}
    \sigma_{ab} X^a Y^b = \left[\left(\gamma \theta_{ab} + \sigma^\Sigma_{ab}\right)
    + \frac{1}{1+2\gamma^2}\left(\gamma \theta_{cd}v^cv^d +\sigma^\Sigma_{cd}v^c v^d\right) h_{ab}\right]X^a Y^b.        
\end{equation}
\end{lemma}

\begin{corollary}\label{cor:shear}
The following relations hold for the shear tensor:
    \begin{align}
    N^a N^b \sigma_{ab} &= \frac{3}{1+2\gamma^2}\left(\gamma\theta_{ab} + \sigma^\Sigma_{ab}\right)v^av^b,\label{eq:shear-E}\\
    -N^a h_b^c \sigma_{ac} &= 
    \frac{v^c}{\gamma} \left(\gamma \theta_{cd}+\sigma_{cd}^\Sigma\right)\left(\delta_b^d + \frac{v^dv_b}{1+2\gamma^2}\right).\label{eq:shear-J}
    \end{align}
\end{corollary}
\begin{proof}
Since $u^a \sigma_{ab} = 0 = u^b \sigma_{ab}$, we can replace $N^a$ with $-v^a/\gamma$
in the contractions on the left-hand sides of equations \eqref{eq:shear-E} and \eqref{eq:shear-J}.
Equation \eqref{eq:shear-E} is then an easy consequence of equation \eqref{eq:shear-final}
and the observation 
\[
1 + \frac{|v|_h^2}{1+2\gamma^2} = \frac{3\gamma^2}{1+2\gamma^2}.
\]
Equation \eqref{eq:shear-J} follows from an analogous direct computation.
\end{proof}

Combining the equations of Lemma \ref{lem:bulk-heat} and Corollary \ref{cor:shear}
we obtain our desired expressions for the energy and momentum densities.
\begin{proposition}
For the viscous conformal fluid stress-energy tensor \eqref{eq:T-vcf},
\begin{align}
\mathcal E &= 
\left(\epsilon+\chi\frac{3\mathcal D\epsilon}{4\epsilon}\right)\left(1+\frac{4}3|v|_h^2\right)
-\eta\frac{6}{1+2\gamma^2}\left[\gamma \theta_{ab}+ \sigma^\Sigma_{ab}\right]v^av^b
+\lambda\frac{v^a}{2\epsilon} \left(\mathcal D_a^\Sigma\epsilon + v_a \mathcal D\epsilon \right)\label{eq:viscous-E}\\
\mathcal J_b &=
\left(\epsilon+\chi\frac34\frac{\mathcal D\epsilon}{\epsilon} \right)\frac{4\gamma}{3}v_b
-\eta\frac{2v^c}{\gamma}\left(\delta_b^d +\frac{v_bv^d}{1+2\gamma^2}\right)\left(\gamma\theta_{cd}+\sigma^\Sigma_{cd}\right)
+\lambda\frac{\gamma}{4\epsilon}\left(\delta^c_b+\frac{v^cv_b}{\gamma^2}\right)
(\mathcal D^\Sigma_c \epsilon + v_c \mathcal D\epsilon).\label{eq:viscous-J}
\end{align}
Alternatively, in terms of temperature,
\begin{align}
    \mathcal E &= 
    \left(c_\epsilon T^4+3c_\chi T^2\, \mathcal DT\right)\left(1+\frac{4}3|v|_h^2\right)
    -c_\eta T^3 \frac{6}{1+2\gamma^2}\left[\gamma \theta_{ab}+ \sigma^\Sigma_{ab}\right]v^av^b
    +2 c_\lambda T^2  v^a \left(\mathcal D_a^\Sigma T + v_a \mathcal DT \right)\label{eq:viscous-E-temp}\\
    \mathcal J_b &=
    \left(c_\epsilon T^4 +3 c_\chi T^2\, {\mathcal DT} \right)\frac{4\gamma}{3}v_b
    -2c_\eta T^3\frac{v^c}{\gamma}\left(\delta_b^d +\frac{v_bv^d}{1+2\gamma^2}\right)\left(\gamma\theta_{cd}+\sigma^\Sigma_{cd}\right)
    +c_\lambda T^2\gamma\left(\delta^c_b+\frac{v^cv_b}{\gamma^2}\right)
    (\mathcal D^\Sigma_c T + v_c \mathcal DT).\label{eq:viscous-J-temp}
    \end{align}
\end{proposition}

\section{Conformal method and matter field scaling}\label{S:Scaling_matter}

The conformal method, summarized briefly below, is a technique for generating solutions
of the constraint equations \eqref{eq:ham-constraint}--\eqref{eq:mom-constraint}.
When applied to a particular matter model, it requires model-specific 
conformal scaling rules, and the bulk of this section is devoted to 
determining appropriate scaling rules for both perfect and viscous conformal fluids.

First consider the vacuum setting. Following the perspective of \cite{maxwell_conformal_2014}, 
the ``seed data'' for the conformal method
consist of a conformal class $\mathbf h$, a momentum $\mathbf p$ 
conjugate to the conformal class (modulo diffeomorphisms),
a mean curvature $\tau$, and a positive volume form $\alpha$.
This description of the seed data is somewhat abstract, and in practice we 
select a representative $h_{ab}$ of $\mathbf h$, in which case $\mathbf p$
can be represented by a trace-free, divergence-free symmetric tensor $p_{ab}$.
Under the conformal transformation $\tilde h_{ab} = h_{ab}$, the 
conjugate momentum transforms in three dimensions as 
$\tilde p_{ab} = e^\phi p_{ab}$.  The mean curvature is a specified function,
and the fixed volume form $\alpha$ is used to define a conformally 
transforming lapse\footnote{We use the notation $N$ for the lapse and $N^a$ for the unit normal to $\Sigma$. Both notations are standard in the literature and the distinction between the two is clear from context.} according to $N = \alpha/ dV_h$ and hence $\widetilde N = e^{3\Omega} N$.  For more on the conformally transforming lapse, 
see \cite{york_conformal_1999}\cite{Pfeiffer:2003ka} in the context of the conformal thin-sandwich formulation
of the conformal method and \cite{maxwell_conformal_2014} for the equivalence 
of the conformal thin sandwich method and the standard conformal method. 

Continuing to work in the vacuum setting, we seek a solution of the constraint equations of the form 
$\tilde h_{ab}= e^{-2\Omega}h_{ab}$ and 
\begin{equation}\label{eq:K-conf}
\begin{aligned}
\widetilde K_{ab} &= \left(\tilde p_{ab} + \frac{1}{2\widetilde N} (\widetilde {\ck} W)_{ab}\right) + \frac{\tau}{3}\tilde h_{ab}\\
& = e^{\Omega}\left( p_{ab} + \frac{1}{2N}(\ck W)_{ab}\right) + \frac{\tau}{3}e^{-2\Omega} h_{ab}
\end{aligned}
\end{equation}
where $e^{\Omega}$ and $W^a$ are an unknown conformal factor and vector field
respectively. In the above, $\ck$ is the conformal Killing operator, 
so $(\ck W)_{ab} = \nabla_a W_b + \nabla_b W_a - (2/3) \nabla_c W^c\, h_{ab}$.
For reasons involving the
conformal transformation of scalar curvature it is convenient to write
$e^\Omega=\phi^{-2}$ for an unknown function $\phi$ instead.
Substituting $\tilde g_{ab}$ and $\widetilde K_{ab}$ into the
constraint equations and rewriting them in terms of the original data
 $(g_{ab}, p_{ab}, \tau, N)$, we find that $\phi$ and $W^a$ solve 
\begin{align}
    -8\Delta\phi + R_h \phi + \left| p+\frac{1}{2N}\ck W\right|_h^2\phi^{-7} + \frac{2}{3}\tau^2\phi^5 &= 0\label{eq:LCBY-ham}\\
    -\nabla^a \left(\frac{1}{2N}(\ck W)_{ab}\right) &= -\frac{2}{3}(d\tau)_b\phi^6.
\end{align}
These are the LCBY equations, named for Lichnerowicz, Choquet-Bruhat, and York.
In particular, if $\tau$ is constant (which we henceforth suppose), the momentum
constraint decouples and no longer involves $\phi$.  Its solutions consist
of the conformal Killing fields satisfying $(\ck W)_{ab}=0$ and have no
impact on the Hamiltonian constraint \eqref{eq:LCBY-ham}, which is known in this 
context as the Lichnerowicz equation and which can be analyzed separately.

Now suppose in the non-vacuum case that the energy density $\mathcal E$ and momentum density 
$\mathcal J_a$ depend on the metric and
on matter fields $\mathcal F$ which can themselves be 
written in terms of matter seed data $F$ and the conformal factor $\phi$.
Hence, $\widetilde {\mathcal E} = \mathcal E( \mathcal F(F, \phi), \phi^4 h_{ab} )$
and similarly for the momentum density.  Assuming $\tau$ is constant,
the LCBY equations then become
\begin{align}
    -8\Delta\phi + R_h \phi + \left| p+\frac{1}{2N}\ck W\right|_h^2\phi^{-7} + \frac{2}{3}\tau^2\phi^5 &= 2\phi^5 \widetilde {\mathcal E}\label{eq:LCBY-ham-matter}\\
    -\nabla^a \left(\frac{1}{2N}(\ck W)_{ab}\right) &= \phi^6 \widetilde{\mathcal J}_b\label{eq:LCBY-mom-matter}
\end{align}
where $\widetilde{\mathcal E}$ and $\widetilde{\mathcal J}_a$ are expressions
involving $h_{ab}$, the matter seed fields $F$, and $\phi$. If $\mathcal J_a$
depends on $\phi$ such that  
\begin{equation}\label{eq:mom-scale}
    \widetilde{\mathcal J}_a = \phi^{-6} \mathcal J_{a}  
\end{equation}
then the momentum constraint
again decouples from the Hamiltonian constraint,
and the analysis reduces to studying the scalar Lichnerowicz equation
just as in the vacuum case.

Remarkably, one can arrange for equation \eqref{eq:mom-scale} to hold 
for a wide range of matter models. 
Early work in this direction appears in \cite{kuchar_kinematics_1976}\cite{isenberg_extension_1977}, which gives a physical argument of 
a general principle that would ensure
its satisfaction.  This principle is further developed and made
rigorous in \cite{IsenbergMaxwellMatter}, where it is shown that if
\vspace{-\topsep}
\begin{itemize}\setlength{\itemsep}{0pt}
    \item the matter model arises from a Lagrangian,
    \item the Lagrangian is diffeomorphism invariant,
    \item the metric appears only algebraically (i.e., is minimally coupled to matter), and
    \item matter is expressed on $\Sigma$ in terms of certain 
    fields and conjugate momenta that are held conformally invariant,
\end{itemize}  
\vspace{-\topsep}
then equation \eqref{eq:mom-scale} is guaranteed to hold.

Although we do not have a Lagrangian description of the viscous fluid model,
and because such a Lagrangian would necessarily involve first derivatives 
of the metric, we cannot apply \cite{IsenbergMaxwellMatter} directly to the full conformal viscous fluid model.  
Instead, we first treat perfect conformal fluids where these techniques can be used.
Having obtained principled scaling relations for some of the matter fields, 
we then extend these by ad-hoc but well-motivated rules to the remaining fields in the viscous case.

\subsection{Perfect conformal fluids}\label{secsec:scaling-perfect}

Using the technique of \cite{IsenbergMaxwellMatter} summarized above,
we show that under the conformal change 
$\tilde h_{ab} = \phi^4 h_{ab} = e^{-2\Omega} h_{ab}$,
the rest number density $n$, rest energy density $\epsilon$, and 
slice projected velocity $v^a$ transform via
\begin{align}
\tilde n &= e^{3\Omega} n\label{eq:conf-n}\\
\tilde \epsilon &= e^{4\Omega} \epsilon\label{eq:conf-epsilon}\\
\tilde v^a &= e^{\Omega} v^a\label{eq:conf-v}.
\end{align}
Equations \eqref{eq:conf-n}--\eqref{eq:conf-v}
form the foundation for the full set of transformations needed for 
the viscous model and
could reasonably be motivated heuristically.  However, since the 
technique of \cite{IsenbergMaxwellMatter} has not yet been applied 
to perfect fluids beyond a couple of sample constitutive relations treated in 
\cite{IsenbergMaxwellMatter}, 
we take this opportunity to further develop it here.
Indeed, while perhaps intuitive, we find that
the transformations \eqref{eq:conf-n}--\eqref{eq:conf-v} are
in fact very specific to the constitutive equation 
of a conformal fluid. 

To begin we require a Lagrangian description of the matter model
and we use the approach for perfect fluids appearing in
\cite{christodulu_action_2000}.  Because this Lagrangian approach 
is perhaps not widely known, we summarize it here; 
see also \cite{IsenbergMaxwellMatter} Section 7.1 for more details.

Individual clusters of fluid particles are tracked in this model using a reference
3-manifold $\overline{\Sigma}$ that is different
from spacetime and which effectively serves as
a set of labels
for the particles. For simplicity we assume that $\overline{\Sigma}$ is oriented,
and we endow $\overline\Sigma$ with a positive 3-form
$\overline \omega$ with the interpretation that 
$\int_\Omega \overline \omega$ is the number of
particles contained in the region 
$\Omega \subseteq \overline \Sigma$.  The pair $(\overline \Sigma, \overline \omega)$
is called the \textit{material manifold}.

A fluid configuration on $M$ is 
a map $\Psi: M \to \overline \Sigma$ and is the
primary fluid variable for the Lagrangian theory.  Note that $M$
has not yet been assigned a Lorentzian structure because the metric
is an independent variable in the Lagrangian and indeed finding a
suitable metric-free description of matter is a key step in
formulating the Lagrangian.
If $x\in \overline \Sigma$ is a fluid element, the 
preimage $\Psi^{-1}(x)$ is the worldline of $x$ in
$M$.  The 3-form $\omega = \Psi^*\overline \omega$
describes the flux of particles through the spacetime,
and one readily verifies $d\omega =0 $, a manifestation of 
local conservation of particles.  This description of a 
fluid configuration assumes that the fluid fills all of
spacetime, a simplifying hypothesis we continue to make.
We further assume that $\Psi$ is full rank
and hence $\Psi_*$ has a one-dimensional kernel which
is tangent to the fluid element worldlines.

If $M$ is also equipped with a time-oriented Lorenzian metric $g_{ab}$
we can describe the particle flux more familiarly 
in terms of the number flux vector $j^a$ defined by
$j\interior dV_g = \omega$.  We assume that $\Psi$ and $g_{ab}$
are compatible in the sense that $j^a$ is time-like and future pointing,
which is an open condition on the Lagrangian fields $g_{ab}$ and $\Psi$. 
The rest number density is $n = \sqrt{-g_{ab}j^a j^b}$ 
whereas the fluid velocity is $u^a = j^a / n$. 
We can also write $n = ||\omega||_g = ||\Psi^* \overline \omega||$
which makes it transparent that 
$n$ is a function of $\Psi$ and $g_{ab}$ that depends algebraically on $g_{ab}$
and that depends on the values and first derivatives of $\Psi$, an important
criterion for the matter Lagrangians treated in \cite{IsenbergMaxwellMatter}.

The equation of state of a perfect fluid 
relates rest number density $n$ to rest energy density
$\epsilon(n)$ and the spacetime fluid Lagrangian is
\[
\mathcal L_{\mathrm{fluid}}[g, \Psi] = -2 \epsilon(n[g, \Psi])\; dV_g.
\] 
As shown in
\cite{IsenbergMaxwellMatter},
the stress-energy tensor associated with this Lagrangian
is
\begin{equation}\label{eq:fluid-T}
T_{ab} = \epsilon(n) u_a u_b + (n\epsilon'(n)-\epsilon(n))(u_a u_b + g_{ab})
\end{equation}
which is the stress-energy tensor of a perfect fluid with
energy density $\epsilon(n)$ and pressure 
$p(n)=n\epsilon'(n)-\epsilon(n)$.

Now consider a Cauchy surface $\Sigma$ with future pointing 
unit normal $N^a$ and induced metric $h_{ab}$.
As usual, we decompose $u^a = \gamma N^a + v^a$ where
$v^a$ is tangent to $\Sigma$ and $\gamma^2 = 1 + |v|_h^2$.  
The map $\Psi$ restricts to a map $\Psi_\Sigma:\Sigma\to \overline\Sigma$
and we obtain a number density $\omega_\Sigma = \Psi_\Sigma^* \overline \omega$
on $\Sigma$. In terms of the metric, we can write $\omega_\Sigma = n_\Sigma dV_h$
for some scalar function $n_\Sigma$ and indeed 
$n_\Sigma = n \gamma$.  A brief computation shows that
the energy and momentum densities associated with 
the stress-energy tensor \eqref{eq:fluid-T} are
\begin{align}
\mathcal E &= \epsilon(n) + n\epsilon'(n) |v|_h^2 \label{eq:ham-fluid-v1}\\
\mathcal J_b &=  \gamma n  \epsilon'(n) v_b \label{eq:mom-fluid-v1}.
\end{align}
Equations \eqref{eq:ham-fluid-v1}--\eqref{eq:mom-fluid-v1} parameterize
these densities in terms of $v^a$ and $n$, which are physically reasonable variables,
but the application of the conformal method requires instead the Lagrangian
field variable (the map $\Psi_\Sigma$) and its conjugate momentum, which
are treated as conformally invariant objects. Hence we need to rewrite these
equations in terms of the new variables.

We can simplify matters somewhat by assuming
(without loss of generality) that the material manifold
$\overline \Sigma$ is the Cauchy surface $\Sigma$ itself,
that $\overline \omega$ is the observed density $\omega_\Sigma$, 
and that $\Psi_\Sigma = \id$.  In this case,
a computation in \cite{IsenbergMaxwellMatter}
shows that the momentum conjugate to $\Psi_\Sigma$ is
the covector-valued three-form
\[
\Pi_a = -2\epsilon'(n)v_a\; \omega_\Sigma.
\]
The conformally invariant data describing the fluid
can be taken to be the 3-form $\omega_\Sigma$ along with
$\Pi_a$.  However, rather than work with $\Pi_a$ directly
we introduce the associated 1-form
\[
p_a = \epsilon'(n)v_b
\]
which can be interpreted as the average momentum per particle
and which is also conformally invariant if $\omega_\Sigma$ and $\Pi_a$ are.
Rewriting equations \eqref{eq:ham-fluid-v1} and \eqref{eq:mom-fluid-v1}
in terms of $\omega_\Sigma$ and $p_a$ we obtain
\begin{align}
    \mathcal E &= \epsilon(n) + \frac{n}{\epsilon'(n)} |p|_h^2\label{eq:cf-ham}\\
    \mathcal J_b &= n_\Sigma p_b\label{eq:cf-mom}
\end{align}
where $n_\Sigma$ is shorthand for $\omega_\Sigma/dV_h$ 
and where $n$ is determined by $n_\Sigma^2 = n^2 \gamma^2$,
i.e.,
\begin{equation}\label{eq:implicit-n-n}
n_\Sigma^2 = n^2\left(1 + \frac{1}{\epsilon'(n)^2} |p|_h^2\right).
\end{equation}

At this point we can determine how $\mathcal E$ and $\mathcal J_a$ 
conformally transform when $h_{ab}$ becomes $\tilde h_{ab}= e^{-2\Omega} h_{ab}$
and when $\omega_\Sigma$ and $p_a$ are held constant.  The rule for the 
momentum density is easy: since $\tilde n_\Sigma dV_{\tilde h}$ and
$n_\Sigma dV_h$ both equal $\omega_\Sigma$
and since $dV_{\tilde h} = e^{-3\Omega} dV_h$ it follows that 
$\tilde n_\Sigma = e^{3\Omega} n_\Sigma$. Since $p_a$ is conformally invariant
we find 
\begin{equation}\label{eq:conf-J-perfect}
    \widetilde {\mathcal J}_a = e^{3\Omega} \mathcal J_a
\end{equation} 
which, recalling the convention $e^\Omega = \phi^{-2}$, is exactly the scaling 
\eqref{eq:mom-scale} needed to achieve decoupling of the momentum constraint 
for the conformal method.  This is to be expected, of course, as it 
is guaranteed to hold when using the \cite{IsenbergMaxwellMatter} approach to matter field scaling.

The conformal transformation rule for $\mathcal E$ requires 
determining the associated rule for $n$, which stems from the implicit
relation \eqref{eq:implicit-n-n}.  Substituting 
the relations 
$\tilde n_\Sigma = e^{3\Omega} n_\Sigma$
and $|p|_{\tilde h}^2 = e^{2\Omega}|p|_h^2$ already derived 
into equation \eqref{eq:implicit-n-n} we find
\begin{equation}\label{eq:implicit-n-n-conf}
e^{6\Omega} n_{\Sigma}^2 = \tilde n^2\left( 1 + \frac{1}{(\epsilon'(\tilde n))^2}e^{2\Omega}|p|_h^2\right)
\end{equation}
which defines $\tilde n$ in terms of $n_\Sigma$, $p_a$ and $e^\Omega$.
The implicit relation \eqref{eq:implicit-n-n-conf} 
complicates the conformal method when applied to an arbitrary equation of state,
and a general analysis is work in progress  \cite{AllenMaxwell}.
Nevertheless, for some specific equations of state one can make further 
progress. Dust and stiff fluids are treated in \cite{IsenbergMaxwellMatter}, for example, and 
we show now that the conformal fluid equation of state also
admits a direct analysis.

For a conformal fluid, $\epsilon(n) = k_\epsilon n^{4/3}$ 
for some constant $k_\epsilon$
and hence equation \eqref{eq:implicit-n-n} becomes
\begin{equation}\label{eq:implicit-n-n-cpf}
n_\Sigma^2 = n^2 + \frac{9}{16k_\epsilon^2}n^{4/3} |p|_h^2.
\end{equation}
Note that monotonicity implies $n\ge 0$ is uniquely determined
by $n_\Sigma$, $p$ and $h$ in equation \eqref{eq:implicit-n-n-cpf}. 
Equation \eqref{eq:implicit-n-n-conf} becomes
\[
e^{6\Omega} n_\Sigma = \tilde n^2 + \frac{9}{16k_\epsilon^2}e^{2\Omega}\tilde n^{4/3}|p|_h^2
\]
or equivalently
\begin{equation}\label{eq:implicit-n-n-cpf2}
n_\Sigma = \left(e^{-3\Omega} \tilde n\right)^2 + 
\frac{3}{4k_\epsilon}\left(e^{-3\Omega} \tilde n\right)^{4/3}|p|_h^2.
\end{equation}
Comparing equations \eqref{eq:implicit-n-n-cpf} and \eqref{eq:implicit-n-n-cpf2},
the uniqueness of $n\ge 0$ satisfying equation \eqref{eq:implicit-n-n-cpf} implies
$\tilde n = e^{3\Omega} n$, which is the first of the relations 
\eqref{eq:conf-n}--\eqref{eq:conf-v}.  But now
\[
\tilde\epsilon = \epsilon(\tilde n) = k_\epsilon \tilde n^{4/3}= e^{4\Omega} cn^{4/3} = e^{4\Omega} \epsilon(n)
\]
which is equation \eqref{eq:conf-epsilon}.
Also, because $p_a = \epsilon'(n)v_b$ is conformally invariant,
\[
\epsilon'(\tilde n) \tilde v_a = \epsilon'(n) v^a.
\]
Since $\epsilon'(\tilde n) = e^\Omega \epsilon'(n)$ we conclude
$\tilde v_a = e^{-\Omega} v_a$ and therefore 
\[
\tilde v^a = \tilde h^{ab} \tilde v_b = e^{2\Omega}h^{ab} e^{-\Omega} v_b = e^{\Omega} v^b
\]
which is equation \eqref{eq:conf-v}.

The conformal transformation rule for $\mathcal E$ can now be obtained by
substituting the constitutive equation $\epsilon(n) = k_\epsilon n^{4/3}$
into equation \eqref{eq:ham-fluid-v1} and conformally transforming $n$ and $h_{ab}$
appropriately.  In fact, with this choice of $\epsilon(n)$, it is easy to see that
\[
\mathcal E = \epsilon(n) \left(1+\frac{4}{3}|v|_h^2\right)
\]
which is exactly equation \eqref{eq:viscous-E} with all viscosity terms set to zero.
Since $|\tilde v|_{\tilde h} = |v|_h$, 
\begin{align}\label{eq:conf-E-perfect}
\widetilde{\mathcal E} &= \tilde \epsilon\left(1+\frac{4}{3}|v|_h^2\right)\\
& = e^{4\Omega} \epsilon\left(1+\frac{4}{3}|v|_h^2\right)\\
& = e^{4\Omega} \mathcal E.
\end{align}

The transformation rules $\tilde n = e^{3\Omega} n$ and 
$\widetilde{\mathcal E} = e^{4\Omega}\mathcal E$ are remarkably
simple, and we emphasize that they are specific to the conformal 
perfect fluid.  By contrast, the analogous relations for 
dust from \cite{IsenbergMaxwellMatter} are
\begin{align*}
\tilde n &= \frac{e^{3\Omega}n_\Sigma}{\sqrt{1+e^{2\Omega}\frac{1}{m^2}|p|_h^2}}\\
\widetilde{\mathcal E} &= 
n_\Sigma e^{4\Omega} \sqrt{e^{-2\Omega} m^2 + |p|_h^2}
\end{align*}
where $m$ is the mass per particle.

\subsection{Viscous conformal fluids}\label{secsec:scaling-viscous}

We now present heuristic arguments to extend 
the scaling rules \eqref{eq:conf-n}--\eqref{eq:conf-v}
to the viscous case.  Specifically, we derive the additional transformations
for $\mathcal D\epsilon$ (or equivalently $\mathcal D T$) and $\mathcal A^\Sigma_a$
along with associated rules for the energy and momentum densities of a viscous conformal fluid.

The scaling $\tilde v^a = e^\Omega v^a$ when
$\tilde h_{ab} = e^{-2\Omega} h_{ab}$ from equation \eqref{eq:conf-v}
is consistent with a spacetime transformation $\tilde u^a = e^\Omega u^a$
when $\tilde g_{ab} = e^{-2\Omega} g_{ab}$. This transformation rule for $u^a$
leaves the Weyl derivative $\mathbfcal D_a$ unchanged: with respect to $g_{ab}$ it
is expressed with the one form $\mathcal A_a$ from equation 
\eqref{eq:A-def} and with respect to $\tilde g_{ab}$ it is expressed 
with $\widetilde{\mathcal A}_a = \mathcal A_a - (d\Omega)_a$.  The induced
Weyl derivative $\mathbfcal D^\Sigma_a$ on $\Sigma$ is then also unchanged, 
which leads to
\begin{equation}\label{eq:conf-A}
\widetilde{\mathcal A}^{\;\Sigma}_a = \mathcal A^\Sigma_a - (d^\Sigma \Omega)_a
\end{equation}
which we take as one of our remaining needed transformation rules.

If the Weyl derivative $\mathbfcal D_a$ 
is fixed and if $\epsilon$ transforms as $\tilde \epsilon = e^{4\Omega}\epsilon$
then $\mathcal D_a \epsilon$ also has weight 4, so 
$\widetilde {\mathcal D}_a \tilde \epsilon = e^{4\Omega} \mathcal D_a \epsilon$.
Moreover, $\widetilde{\mathcal D}\tilde \epsilon =
\tilde u^a \widetilde{\mathcal D}_a \tilde \epsilon = e^{\Omega}
u^a e^{4\Omega}\mathcal D_a \epsilon$.  Hence we obtain our final 
conformal transformation
\begin{equation}
    \label{eq:conf-De}
    \widetilde{\mathcal D \epsilon} = e^{5\Omega} \mathcal D\epsilon.
\end{equation}

In summary, our scaling rules for the conformal viscous fluid matter fields under the metric
transformation $\tilde h_{ab} = e^{-2\Omega} h_{ab}$
are 
\begin{equation}\label{eq:conf-matter}
\begin{aligned}
\tilde \epsilon &= e^{4\Omega} \epsilon\\
\tilde v^a &= e^\Omega v^a\\
\widetilde{\mathcal A}^{\,\Sigma}_a &= \mathcal A^\Sigma_a - (d^\Sigma \Omega)_a\\
\widetilde{\mathcal D\epsilon} &= e^{5\Omega} \mathcal D\epsilon.
\end{aligned}
\end{equation}

Although matter scaling for the conformal method is not generally 
associated with a spacetime conformal scaling, for the viscous 
conformal fluid it is.  Indeed the matter transformations 
\eqref{eq:conf-matter}  all follow from spacetime scalings
\begin{equation}\label{eq:spacetime-conf}
    \begin{aligned}
        \tilde g_{ab} &=  e^{-2\Omega} g_{ab}\\
        \tilde u^a &= e^\Omega u^a\\
        \tilde \epsilon &= e^{4\Omega}\epsilon.
    \end{aligned}            
\end{equation}
A straightforward computation based on relations \eqref{eq:spacetime-conf} along with the scaling rule \eqref{eq:Weyl-conf} for the Weyl derivative 
shows that the stress-energy tensor then transforms as 
\begin{equation}\label{eq:T-conf}
\begin{aligned}
\widetilde T_{ab} &= e^{2\Omega} T_{ab}.
\end{aligned}
\end{equation}
This is precisely the scaling relation needed to ensure equation \eqref{E:Conformal_transformation_div_T} when $T_{ab}$ is trace-free and hence
$\widetilde \nabla^a \widetilde T_{ab} = 0$
when $\nabla^a T_{ab}=0$, a hallmark
of a conformally invariant matter field.

Since $\epsilon$ is proportional to $T^4$, equations \eqref{eq:conf-matter}
imply associated relations for temperature,
\begin{equation}\label{eq:conf-temp}
\begin{aligned}
    \widetilde T &= e^\Omega T\\
    \widetilde{\mathcal D T} &= e^{2\Omega} \mathcal D T.
\end{aligned}
\end{equation}
Because $v_a$ has weight $-1$, 
the induced shear tensor $\sigma^\Sigma_{ab}$ from
equation \eqref{eq:shear-etc} satisfies
\begin{equation}\label{eq:conf-shear}
\widetilde \sigma_{ab} = e^{-\Omega} \sigma_{ab}.
\end{equation}

With these elementary conformal scalings in hand we can compute
how the energy density transforms. First, we decompose
$\mathcal E$ from equation \eqref{eq:viscous-E-temp} as
\[
\mathcal E = \mathcal E^{\mathrm{cpf}} + c_\chi \mathcal E^{\chi} + c_\eta \mathcal E^{\eta} + c_\lambda \mathcal E^{\lambda} -c_\eta \frac{6\gamma T^3}{1+2\gamma^2}\theta_{ab}v^a v^b
\]
where
\begin{equation}\label{eq:E-parts}
    \begin{multlined}
\mathcal E^{\mathrm{cpf}} = c_\epsilon T^4 \left(1+\frac{4}{3}|v|_h^2\right);\quad
\mathcal E^{\mathrm{\chi}} =  3 T^2\, \mathcal DT
\left(1+\frac{4}{3}|v|_h^2\right);\quad
\mathcal E^{\eta} = -\frac{6T^3}{1+2\gamma^2}\sigma^\Sigma_{ab}v^a v^b;\quad\\
\mathcal E^{\lambda} = 2T^2\gamma v^a\left(\mathcal D_a^\Sigma T + v_a \mathcal DT\right).    \hfil
    \end{multlined}
\end{equation}
Since $|v|_h^2$ and $\gamma$ are conformally invariant, direct
computation based on equations \eqref{eq:conf-matter}--\eqref{eq:conf-shear} shows 
that each of 
$\mathcal E^{\mathrm{cpf}}$, $\mathcal E^{\mathrm{\chi}}$,
$\mathcal E^{\mathrm{\eta}}$ and $\mathcal E^{\mathrm{\lambda}}$
scale with a factor $e^{4\Omega}$ and hence
\begin{equation}\label{eq:E-conf}
\widetilde{\mathcal E} = e^{4\Omega}\left(
    \mathcal E^{\mathrm{cpf}} + c_\chi \mathcal E^{\chi} + c_\eta \mathcal E^{\eta} + c_\lambda \mathcal E^{\lambda}
\right) -c_\eta e^{5\Omega}\frac{6\gamma T^3} {1+2\gamma^2}\tilde \theta_{ab}v^a v^b
\end{equation}
where $\tilde \theta_{ab}$ is, for the moment, not yet specified.

Similarly, we split the momentum density $\mathcal J_a$ from equation \eqref{eq:viscous-J-temp}
as
\[
\mathcal J_a = \mathcal J_a^{\mathrm{cpf}} + c_\chi \mathcal J_a^{\chi}
+c_\eta \mathcal J_a^{\eta} + c_\lambda \mathcal J_a^\lambda -2 c_\eta T^3\left(\delta_a^d + \frac{v_a v^d}{1+2\gamma^2}\right)v^c\theta_{cd}
\]
with
\begin{equation}\label{eq:J-parts}
    \begin{multlined}
\mathcal J_a^{\mathrm{cpf}} = c_\epsilon T^4 \frac{4\gamma}{3}v_a;\quad
\mathcal J_a^{\chi}=4 \gamma T^2\, \mathcal D T\; v_a;\quad 
\mathcal J_a^{\eta}=-\frac{2T^3}{\gamma}\left(\delta_b^d+\frac{v_bv^d}{1+2\gamma^2}\right)v^c \sigma^\Sigma_{cd};\\
\mathcal J_a^{\lambda}=\gamma T^2 \left(\delta^c_a+\frac{v^cv_a}{\gamma^2}\right)\left(\mathcal D_c^\Sigma T + v_c\mathcal DT\right).\hfil
    \end{multlined}
\end{equation}
Again using equations \eqref{eq:conf-matter}--\eqref{eq:conf-shear} 
we find that each of $\mathcal J_a^{\mathrm{cpf}}$,
$\mathcal J_a^{\chi}$, $\mathcal J_a^{\eta}$ and $\mathcal J_a^{\lambda}$, 
scale with a factor of $e^{3\Omega}$ and hence
\begin{equation}\label{eq:J-conf}
\widetilde {\mathcal J}_a = e^{3\Omega}
\left( 
    \mathcal J_a^{\mathrm{cpf}} + c_\chi \mathcal J_a^{\chi}
    +c_\eta \mathcal J_a^{\eta} + c_\lambda \mathcal J_a^\lambda    
\right) -2c_\eta e^{4\Omega} T^3
\left(\delta_a^d + \frac{v_a v^d}{1+2\gamma^2}\right)v^c\tilde \theta_{cd}.
\end{equation}

At this point we encounter a tension between the conformal method
and a spacetime interpretation of rescaling matter fields.
From a spacetime perspective, the trace-free part of the 
second fundamental form of $\Sigma$ becomes
$e^{-\Omega} \theta_{ab}$ under the conformal change $g_{ab} \to e^{-2\Omega}g_{ab}$.  
If we were able to adopt this transformation for $\tilde \theta_{ab}$, 
then $\mathcal E$ and $\mathcal J_a$
would transform exactly as for a perfect conformal fluid, equations 
\eqref{eq:conf-J-perfect} and \eqref{eq:conf-E-perfect}, and the 
application of the conformal method to the viscous fluid would proceed exactly 
as in the inviscid case, Lemma \ref{lem:cpf-data} below.  

By contrast, the scaling of $\theta_{ab}$ in the conformal method arises
through its association with the momentum conjugate to the conformal class
of the metric \cite{maxwell_conformal_2014}. From equation \eqref{eq:K-conf}
one computes instead
\begin{equation}\label{eq:theta-conf}
\tilde \theta_{ab} = e^{\Omega} \left( p_{ab} + \frac{1}{2N} (\ck W)_{ab}\right)
\end{equation}
with the lapse $N$, conformal momentum $p_{ab}$ and vector field $W^a$ discussed
in Section \ref{S:Scaling_matter}.
We take equations \eqref{eq:E-conf} and \eqref{eq:J-conf},
supplemented by equation \eqref{eq:theta-conf}, as the conformal transformation rules
for the energy and momentum density.

The failure of $\theta_{ab}$ to conveniently scale keeps the momentum constraint
from decoupling from the Hamiltonian constraint for the viscous fluid; see equation
\eqref{eq:LCBY-mom-vcf} below. Although we are not able to address this issue,
we point to a direction that might lead to better seed data.  The tensor $p_{ab}$
represents, in standard scenarios, a momentum conjugate to the conformal class of the metric.
In the current setting, we do not have a Lagrangian at hand from which one can compute momenta.
However, if one could be found, the non-gravitational part
would depend on the first derivatives of the metric. Hence the momentum represented 
by $p_{ab}$ would require adjusting in the same way that the momentum of a charged
particle involves the electromagnetic vector potential in addition to the free particle momentum.  
Usage of the this true momentum,
rather than $p_{ab}$, may lead to a better system of equations that we have not been able to identify.

\section{Initial data construction}\label{S:Initial_data}

In this section we solve the LCBY equations \eqref{eq:LCBY-ham-matter}--\eqref{eq:LCBY-mom-matter} starting from the following seed fields:
\begin{equation}\label{eq:seed-data}
    \begin{aligned}
    \text{metric:}\quad& h_{ab},\\
    \text{conformal class momentum:}\quad& p_{ab},\\
    \text{mean curvature:}\quad&\tau\in \Reals,\\
    \text{lapse:}\quad& N,\\
    \text{temperature:}\quad& T,\\
    \text{projected velocity:}\quad& v^a,\\
    \text{along-flow temperature derivative:}\quad& \mathcal DT,\\
    \text{induced Weyl connection:}\quad& \mathcal A^\Sigma_a.
    \end{aligned}
\end{equation}
The physical energy and momentum densities 
$\widetilde{\mathcal E}$ and $\widetilde{\mathcal J}_a$
appearing on the right-hand sides of
equations \eqref{eq:LCBY-ham-matter}--\eqref{eq:LCBY-mom-matter}
are obtained by substituting the matter field scaling relations
\eqref{eq:conf-matter}--\eqref{eq:conf-temp} into
the energy and momentum densities
 \eqref{eq:viscous-E-temp}--\eqref{eq:viscous-J-temp}.
Doing so, the LCBY equations become
\begin{equation}\label{eq:LCBY-Ham-vcf}
    \begin{multlined}
        -8\Delta \phi + R_h \phi -\left|p+\frac{1}{2N}\ck W\right|_h^2\phi^{-7}+\frac{2}{3}\tau^2\phi^5 
         =\\
         \qquad\qquad\qquad   2\phi^{-3}\left[\mathcal E^{\rm cpf}  + c_\chi \mathcal E^\chi + c_\eta \mathcal E^{\eta} +c_\lambda \mathcal E^{\lambda}\right]
        -
        2c_\eta\phi^{-7}\frac{6\gamma T^3}{1+2\gamma^2}\left[p(v,v) + \frac{1}{2N}(\ck W)(v,v)\right]
    \end{multlined}\\
    \end{equation}
    and 
    \begin{equation}\label{eq:LCBY-mom-vcf}
    \begin{multlined}
            -\nabla^a \left(\frac{1}{2N} (\ck W)_{ab}\right) =\\
            \qquad\qquad\mathcal J^{{\rm cpf}}_b + c_\chi \mathcal J^\chi_b 
            + c_\eta \mathcal J^{\eta}_{b} 
            + c_\lambda \mathcal J^\lambda_b
            -2c_\eta \phi^{-4} T^3\left(\delta^d_b+\frac{v^d v_b}{1+2\gamma^2}\right)v^c\left(p_{cd}
            +\frac{1}{2N}(\ck W)_{cd}\right)
        \end{multlined}    
\end{equation}
where $\gamma= 1+|v|_h^2$ and where 
equations \eqref{eq:E-parts} and \eqref{eq:J-parts} define
the various terms $\mathcal E^{\rm{cpf}}$, $\mathcal J_b^{\rm{cpf}}$ and so on.

Before proceeding to solve equations \eqref{eq:LCBY-Ham-vcf}--\eqref{eq:LCBY-mom-vcf},
we remark on a conformal covariance property enjoyed by these equations that we use extensively below.
Suppose we start with seed data 
$(h_{ab}, p_{ab}, \tau, N, T, v^a, \mathcal D T, \mathcal A_a^\Sigma)$
but then introduce a conformal change to hatted seed variables $\hat h_{ab}$
and so forth given by
\[
(e^{-2\widehat\Omega} h_{ab}, e^{\widehat \Omega} p_{ab}, \tau, 
e^{-3 \widehat \Omega} N, e^{\widehat \Omega} T, e^{\widehat\Omega} v^a, e^{5\widehat \Omega} \mathcal D T, 
\mathcal A_a^\Sigma - (d^\Sigma\widehat\Omega)_a)
\] 
for some $\widehat \Omega$.
Then $\phi$ and $W^a$ solve equations \eqref{eq:LCBY-Ham-vcf}--\eqref{eq:LCBY-mom-vcf}
with respect to the original seed data if and only if 
$\widehat \phi =e^{\Omega/2} \phi$ and $\widehat W^a = W^a$
solve the same equations with respect to the hatted seed data.
This fact follows from an elementary computation using the various scaling rules and 
critically uses the conformally transforming lapse $\widehat N = e^{-3\widehat \Omega}N$ discussed in Section \ref{S:Scaling_matter}.
As a consequence of this conformal covariance, when solving
equations \eqref{eq:LCBY-Ham-vcf}--\eqref{eq:LCBY-mom-vcf}
we are welcome to first make a conformal 
transformation of the data to start with any convenient representative 
of the conformal class of $h_{ab}$.

\subsection{Compact manifolds}\label{secsec:initial-data-compact}

Suppose throughout this section 
that $\Sigma$ is a compact, connected $3$-manifold. 
We work with standard $L^2$-based Sobolev spaces $H^k(\Sigma)$ with $k\in \Ints_{\ge 0}$
and assume the following regularity
for the seed data:
\begin{equation}\label{eq:data-set}
\begin{aligned}
h_{ab} &\in H^k(\Sigma),\\
p_{ab} &\in H^{k-1}(\Sigma),\\
\tau &\in \Reals,\\
N &\in H^{k}_+(\Sigma)=\{u\in H^k(\Sigma): u>0\},\\
T  &\in H^{k-1}(\Sigma),\\
v^a  &\in H^{k-1}(\Sigma),\\
\mathcal DT &\in H^{k-2}(\Sigma),\\
\mathcal A^\Sigma_a  &\in H^{k-2}(\Sigma).
\end{aligned}
\end{equation}
The conformal class momentum $p_{ab}$ is trace-free and divergence-free and can be
constructed in the usual way using York splitting 
\cite{york_conformally_1973} via \cite{maxwell_rough_2005} Theorem 5.2. 
We assume additionally that $N>0$ and that $T\ge 0$; regions where $T=0$
are devoid of fluid. A collection \eqref{eq:data-set} 
satisfying these properties is a 
\textbf{conformal fluid seed data set} with regularity $k$.
For perfect conformal fluids,
we omit the irrelevant data $\mathcal D T$ and $\mathcal A^\Sigma_a$
and refer to the collection as a \textbf{perfect conformal fluid seed data set}.

In practice, we  assume $k\ge 3$ and hence $H^{k}(\Sigma)$ and $H^{k-1}(\Sigma)$
are both algebras and multiplication $H^{k-1}(\Sigma)\times H^{k-2}(\Sigma)\to H^{k-2}(\Sigma)$
is well defined.  Using these observations, a computation shows each of the terms
$\mathcal E^{\mathrm{cpf}}$,
$\mathcal J^{\mathrm{cpf}}_a$ and so forth on the right-hand sides
of equations \eqref{eq:LCBY-Ham-vcf}--\eqref{eq:LCBY-mom-vcf}
are well-defined elements of $H^{k-2}(\Sigma)$.

First consider the inviscid case where $c_\chi=c_\eta=c_\lambda=0$.  
Equations \eqref{eq:LCBY-Ham-vcf}--\eqref{eq:LCBY-mom-vcf} become
\begin{align}
    -8\Delta \phi + R_h \phi -\left|p+\frac{1}{2N}\ck W\right|_h^2\phi^{-7}+\frac{2}{3}\tau^2\phi^5 
    &=2\phi^{-3}\mathcal E^{\rm cpf}\label{eq:LCBY-Ham-pcf}\\
    -\nabla^a \left(\frac{1}{2N} (\ck W)_{ab}\right) &=\mathcal J^{{\rm cpf}}_b\label{eq:LCBY-mom-pcf}
\end{align}
and, importantly, the momentum constraint in this system 
can be solved independently of the Hamiltonian constraint.

The solvability of system \eqref{eq:LCBY-Ham-pcf}--\eqref{eq:LCBY-mom-pcf}
depends on the Yamabe invariant $\mathcal Y_g$ of the conformal class
of $h_{ab}$, which for compact manifolds is defined by
\[
  \mathcal Y_h = \inf_{\phi \in C^\infty(\Sigma) \atop \phi\not\equiv 0}  
  \frac{\int_\Sigma 8|\nabla \phi|_h^2 + R\phi^2 dV_h }{||\phi||_{L^6}^2}.
\]
For example, suppose the mean curvature $\tau$ of a solution
of the constraints vanishes and that the energy density is non-negative
and not identically zero.  
The Hamiltonian constraint \eqref{eq:ham-constraint}
then reads $R_h = |K|^2_h + 2\mathcal E$ and consequently the scalar curvature
is non-negative and does not vanish identically.  
This is only possible if the conformal class of $h_{ab}$ is Yamabe positive,
i.e., $\mathcal Y_h>0$.


\begin{lemma} \label{lem:cpf-data}
Consider a perfect conformal fluid seed data set with regularity $k\ge 3$ on a 
connected compact 3-manifold $\Sigma$.
Suppose $h_{ab}$ has no conformal Killing fields, that $T\not\equiv 0$,
and that either $\mathcal Y_h >0$ or $\tau\neq 0$.
Then there exists a unique 
 solution $\phi\in H^{k}(\Sigma)$ and $W^a\in H^k(\Sigma)$
 of system  \eqref{eq:LCBY-Ham-pcf}-\eqref{eq:LCBY-mom-pcf}.
\end{lemma}
\begin{proof}
The conformal covariance of equations \eqref{eq:LCBY-Ham-vcf}--
\eqref{eq:LCBY-mom-vcf} discussed above
allows us to assume without loss of generality
that $h_{ab}$ is any particular 
representative of its conformal class.  From 
\cite{maxwell_rough_2005} Proposition 3.3 we can 
therefore assume that the scalar curvature of
$h_{ab}$ is continuous and has a constant sign
agreeing with the sign of $\mathcal Y_h$.

Because $h_{ab}$ has no conformal Killing fields,
Proposition 5.1 of \cite{maxwell_rough_2005} implies 
the vector Laplacian $\nabla^a (\ck \cdot)_{ab}$ is an isomorphism
$H^k(\Gamma)\to H^{k-2}(\Gamma)$.
The same fact remains true for the lapse-modified operator 
$\nabla^a (1/(2N)\ck \cdot)_{ab}$,
which can be established either by adjusting the proof of
\cite{maxwell_rough_2005} Proposition 5.1 or by simply observing that
$ N\, \nabla^a ( N^{-1}\ck \cdot)_{ab}$ is the 
vector Laplacian of the conformally related metric 
$N^{-2/3} h_{ab}$.  Regardless, 
since $\mathcal J^{\mathrm{cpf}}_a\in H^{k-2}(\Sigma)$, 
equation \eqref{eq:LCBY-mom-pcf} admits a 
unique solution $W^a \in H^k(\Sigma)$.

Having solved for $W^a$, we next show that 
the Lichnerowicz equation \eqref{eq:LCBY-Ham-pcf} has as unique solution
using barrier techniques.
A function $w_+\in H^k(\Sigma)$ 
is a supersolution of equation \eqref{eq:LCBY-Ham-pcf} if
\[
    -8\Delta w_+ + R_h w_+ -\left|p+\frac{1}{2N}\ck W\right|_h^2w_+^{-7}+\frac{2}{3}\tau^2w_+^5 
    \ge 2w_+^{-3}\mathcal E^{\rm cpf}
\]
with the inequality holding pointwise-almost everywhere. 
A subsolution $w_-$ is defined similarly with the inequality reversed.  
Proposition 6.2 of \cite{maxwell_rough_2005} implies that if we can find a 
sub/supersolution pair $w_-$ and $w_+$ with $w_-\le w_+$, then there 
exists a solution $\phi$ of equation \eqref{eq:LCBY-Ham-pcf}
with $w_- \le \phi \le w_+$. 

For brevity of notation let $\alpha = 2/3\, \tau^2$ and let
$\beta = |p+1/(2N)\ck W|^2$; observe that $\beta\in H^{k-1}(\Sigma)$. 
Consider the PDE
\begin{equation}\label{eq:w}
-8\Delta w + (R_+ + \alpha) w  = \beta + 2\mathcal E^{\mathrm{cpf}}
\end{equation}
where $R_+ = \max(R,0)\in H^{k-2}(\Sigma)$.  Because we have assumed
that $\tau\neq 0$
if $Y_h\le 0$, we know that $R_+ + \alpha > 0$ and hence 
\cite{maxwell_rough_2005} Proposition 6.1 implies 
there exists a unique solution $w\in H^k(\Sigma)$ of equation \eqref{eq:w}.
Since each of $\beta$ and $\mathcal E^{\mathrm{cpf}}$ are nonnegative,
and since $\mathcal E^{\mathrm{cpf}}\not\equiv 0$ as a consequence
of our hypothesis $T\not\equiv 0$, we know that $w\not \equiv 0$ as well. 
The maximum principle 
(\cite{maxwell_rough_2005} Lemma 2.9 and Proposition 2.10) 
then implies $w>0$ everywhere.

We claim that $cw$ is a subsolution of equation \eqref{eq:LCBY-Ham-pcf}
if $c$ is sufficiently small, and a supersolution if $c$ is sufficiently large.
A computation using the fact that $R-R_+=R_-:=\min(R,0)$ shows
\begin{equation}\label{eq:sub-sup}
 -8\Delta cw + R_h cw + \alpha(cw)^5 + \beta(cw)^{-7} -2\mathcal E^{\mathrm{cpf}}(cw)^{-3}
 = R_- cw + \alpha((cw)^5-cw) + \beta(c-(cw)^{-7}) + 2\mathcal E^{\mathrm{cpf}}(c-(cw)^{-3}).
\end{equation}
Because $w$ is continuous and positive on the compact manifold $\Sigma$, we can take $c$ sufficiently small so that
each of the expressions $(cw)^5-cw$, $c-(cw)^{-7}$ and $c-(cw)^{-3}$
are negative and hence the last three terms of equation
\eqref{eq:sub-sup} are nonpositive.  The first term is manifestly 
non-positive and hence we can pick a small value $c_->0$
such that $w_-=c_- w$ is a subsolution.
A similar argument shows that if $c$ is sufficiently large, then
each of $c-(cw)^{-7}$ and $c-(cw)^{-3}$ are positive and hence the last
two terms of the right-hand side of 
equation \eqref{eq:sub-sup} are non-negative.  The first two terms 
can be written
\[
\alpha (cw)^5 +(R_--\alpha)cw.
\]
Hence if $c$ is sufficiently large, e.g. 
\[
c > \frac{1}{\min w} \left(\frac{\alpha-\min(R_-)}{\alpha}\right)^{\frac 14},
\] 
then $\alpha (cw)^5 +(R_--\alpha)cw\ge 0$ as well.  Hence
we can pick a large value $c_+>c_-$ such that $w_+=c_+w$ is a supersolution.
Since $w_-\le w_+$ are a sub- and supersolution pair, we obtain a 
solution $\phi$ of equation \eqref{eq:LCBY-Ham-pcf}.

Turning to uniqueness, suppose $\phi_1$ and $\phi_2$ are
two solutions.  Again using conformal covariance,
we can assume that $\phi_2=1$, in which case $\phi_1$
solves
\begin{equation}\label{eq:uniq}
-8\Delta\phi_1 + \alpha(\phi_1^5-\phi_1) + \beta(\phi_1-\phi_1^{-7})    
+2\mathcal E^{\mathrm{cpf}}(\phi_1-\phi_1^{-3})=0.
\end{equation}

Since $(\phi_1-1)_+ = \max(\phi_1-1,0)\in H^1(\Sigma)$, we can apply
integration by parts to conclude
\begin{equation}\label{eq:int-by-parts}
8\int_M |\nabla(\phi_1-1)_+|_h^2\; dV_h = 
\int_M (-8\Delta (\phi_1-1)) (\phi_1-1)_+\; dV_h 
= \int_{U} (-8\Delta \phi_1) (\phi_1-1)_+\; dV_h
\end{equation}
where $U$ is the open set where $\phi_1>1$. On $U$, $-8\Delta \phi_1 \le 0$
and hence $\int_{U} (-8\Delta \phi_1) (\phi_1-1)_+\; dV_h \le 0$. 
But then as a consequence of equation \eqref{eq:int-by-parts}
and the connectivity of $\Sigma$ we 
conclude $(\phi_1-1)_+$ is constant.  A similar argument shows that $(\phi_1-1)_-$, 
defined analogously, is also constant and it follows that $\phi_1$ is itself constant. 
But then equation \eqref{eq:uniq} implies
\[
    \alpha(\phi_1^5-\phi_1) + \beta(\phi_1-\phi_1^{-7})    
    +2\mathcal E^{\mathrm{cpf}}(\phi_1-\phi_1^{-3})=0.
\]
Since $\alpha$, $\beta$ and $\mathcal E^{\mathrm{cpf}}$ are all non-negative,
and since $\mathcal E^{\mathrm{cpf}}$ does not vanish identically, this can only 
happen if $\phi_1\equiv 1$.
\end{proof}

We now show that starting with a perfect conformal fluid solution
of the constraint equations generated by the previous lemma, we can 
perturb the viscosity coefficients to nonzero values
and generate nearby conformal viscous fluid solutions.

\begin{theorem}\label{thm:initial-cpct}
Consider a conformal fluid seed data set with regularity $k\ge 3$ on a compact, connected
3-manifold $\Sigma$.  Suppose additionally
that $h_{ab}$ has no conformal Killing fields,
that $T\not\equiv 0$, and that either $\tau\neq 0$ or $\mathcal Y_h > 0$.

Let $(\phi_{\mathrm{cpf}}, W^a_{\mathrm{cpf}})$ be the unique solution 
of the viscosity-free equations \eqref{eq:LCBY-Ham-pcf}--\eqref{eq:LCBY-mom-pcf}
provided by Lemma \ref{lem:cpf-data} for the corresponding perfect conformal
fluid seed data set.
There is a neighborhood $\Gamma$ of 
$(\phi_{\mathrm{cpf}}, W^a_{\mathrm{cpf}})$ 
in $H^k(\Sigma)\times H^k(\Sigma)$ such that
for any sufficiently small viscosity coefficients 
$c_\chi$, $c_\eta$, $c_\lambda$
there exists a unique solution of system
\eqref{eq:LCBY-Ham-vcf}--\eqref{eq:LCBY-mom-vcf}
in $\Gamma$.
\end{theorem}
\begin{proof}
Let
\begin{multline}\label{eq:FH}
    F_{H}[(c_\chi, c_\lambda, c_\eta),(\phi, W^a) ] = 
    -8\Delta \phi + R\phi  - \left|p+\frac{1}{2N}\ck W\right|^2_h\phi^{-7} 
    +\frac 23 \tau^2 \phi^{5} \\
    - 2\phi^{-3}( \mathcal E^{\mathrm{cpf}}
    +c_\chi \mathcal E^{\chi} + c_\eta \mathcal E^{\eta} + c_\lambda \mathcal E^{\lambda})
-
2c_\eta\phi^{-7}\frac{6\gamma T^3}{1+2\gamma^2}\left[p(v,v) + \frac{1}{2N}(\ck W)(v,v)\right]
\end{multline}
and
\begin{multline}\label{eq:FM}
    F_{M}[(c_\chi, c_\eta, c_\lambda), (\phi, W^a) ] = \\
\nabla^a\left(\frac{1}{2N} (\ck W)_{ab}\right) 
+ \mathcal J_b^{\mathrm{cpf}} + c_\chi \mathcal J_b^{\chi}
+ c_\eta \mathcal J_b^{\eta}
+ c_\lambda \mathcal J_b^{\lambda}
-2c_\eta \phi^{-4} T^3\left(\delta^d_b+\frac{v^d v_b}{1+2\gamma^2}\right)v^c\left(p_{cd}
+\frac{1}{2N}(\ck W)_{cd}\right).
\end{multline}
The function $F=F_H\times F_M$ takes $(c_\chi, c_\eta, c_\lambda)\in \Reals^3$
and $(\phi, W)\in Y := H_+^k(\Sigma)\times H^k(\Sigma)$ to $Z:=H^{k-2}(\Sigma)\times H^{k-2}(\Sigma)$.
In fact, $F$ is continuously differentiable in its arguments, 
a claim that only relies on 
standard properties of multiplication between
Sobolev spaces along with continuous differentiability of the map $\phi\mapsto \phi^{-1}$
on $H^k_+(\Sigma)$. This latter fact can be established by techniques similar to 
those used to prove the more technical result Lemma \ref{lem:one-over} below.

From conformal covariance we can assume without loss of generality 
that $\phi_{\mathrm {cpf}}\equiv 1$ and 
hence $F[0, (1, W^a_{\mathrm{cpf}})]=0$.  Moreover,
at the point 
$(0,(1,W^a_{\mathrm{cpf}}))\in \Reals^3\times Y$
 the derivative of $F$ with respect
to $\phi$ and $W^a$ has the block form
\[
DF = \begin{bmatrix}     
-8\Delta + \Lambda & *\\
0 & \mathrm{div} \left(\frac{1}{2N} \ck \cdot \right)
\end{bmatrix}
\]
with 
\begin{equation}
    \Lambda = R_h +7\left|p+\frac{1}{2N}\ck W\right|^2 +\frac{10}{3}\tau^2+3\mathcal E^{\mathrm{cpf}}
\end{equation}
where $*$ is an irrelevant block. Moreover, because $\phi_{\mathrm{cpf}}$
and $W^a_{\mathrm{cpf}}$ solve the viscosity-free equations, 
and because $\phi_{\mathrm{cpf}}=1$, 
the Hamiltonian constraint implies
\[
R_h = \left|p+\frac{1}{2N}\ck W\right|_h^2 -\frac{2}{3}\tau^2 + 2\mathcal E^{\mathrm{cpf}}
\]
and we conclude
\[
\Lambda = 8\left|p+\frac{1}{2N}\ck W\right|_h^2 +\frac{8}{3}\tau^2+8\mathcal E^{\mathrm{cpf}}.
\]
In particular $\Lambda\ge 0$. Moreover, since $T\not\equiv 0$ 
it follows that $\mathcal E^{\mathrm{cpf}}\not\equiv 0$
and therefore $\Lambda \not\equiv 0$.  Noting that $\Lambda\in H^{k-2}(\Sigma)$,
\cite{maxwell_rough_2005} Proposition 6.1 implies that the operator $-8\Delta + \Lambda$
in the upper-left block of $DF$ is invertible as a map 
$H^{k}(\Sigma)\to H^{k-2}(\Sigma)$.
As in the proof of Lemma \ref{lem:cpf-data}, 
because we have assumed that $h_{ab}$ admits no conformal Killing fields, 
the operator $\div(1/(2N)\ck\cdot )$ in the lower-right block is
also invertible $H^{k}(\Sigma)\to H^{k-2}(\Sigma)$, which establishes the 
invertibility of $DF$ at $(0,(\phi_\mathrm{cpf},W^a_{\mathrm{cpf}}))$.

The Implicit Function Theorem now implies that
there is a small ball $\Gamma$ about $(\phi_\mathrm{cpf}, W^a_{\mathrm{cpf}})$
in $Y$ such that if $(c_\chi, c_\eta, c_\lambda)$
is sufficiently small in $\Reals^3$, there exists a unique $(\phi, W^a)\in \Gamma$
such that $F((c_\chi, c_\eta, c_\lambda),(\phi, W^a))=0$.
\end{proof}

\subsection{Asymptotically Euclidean manifolds}\label{secsec:initial-data-asymptotic}

In this section we extend Theorem \ref{thm:initial-cpct} to isolated 
gravitational systems using weighted Sobolev spaces to enforce
the requisite decay of the various fields.  Following \cite{bartnik_mass_1986},
given $k\in\mathbb{Z}_{\ge 0}$ and $\delta\in \Reals$
the weighted space $H^k_{\delta}(\Reals^3)$ 
consists of the functions $u$
on $\Reals^3$ for which the norm
\[
||u||_{H^k_{\delta}} = \left[ \sum_{|\alpha|\le k} \int_{\Reals^3} |\nabla^\alpha u|^2\; |x|^{-3-2\delta+2|\alpha|}\; dx\right]^{1/2}
\]
is finite; an analogous definition holds for tensors, and we use the same notation. 
The convention for the decay 
parameter $\delta$ is chosen so that
$|x|^\delta \in H^{0}_{\delta-\epsilon}(\Reals^3)$ for any $\epsilon>0$. 
Note that if $u\in H^k_\delta(\Reals^3)$ with $k\ge 1$ then 
its derivatives satisfy $\partial u \in H^{k-1}_{\delta-1}(\Reals^3)$
and lose an order of growth at infinity.
Other properties of weighted Sobolev spaces can be found
in \cite{bartnik_mass_1986}; see also \cite{maxwell_rough_2006} and 
\cite{maxwell_solutions_2005} 
which contains PDE tools used in the construction below.  

We are primarily
interested in the case $k\ge 3$, in which case elements of $H^k(\Reals^3)$
are $C^1$.  Moreover, if $k\ge 3$ and $\delta_1$, $\delta_2\in\Reals$,
then pointwise multiplication determines continuous maps as follows:
\begin{equation}\label{eq:mult}
\begin{aligned}
    H^k_{\delta_1}(\Reals^3)\times H^k_{\delta_2}(\Reals^3) &\to H^k_{\delta_1+\delta_2}(\Reals^3)\\
    H^{k-1}_{\delta_1}(\Reals^3)\times H^{k-1}_{\delta_2}(\Reals^3) &\to H^{k-1}_{\delta_1+\delta_2}(\Reals^3)\\
    H^{k-1}_{\delta_1}(\Reals^3)\times H^{k-2}_{\delta_2}(\Reals^3) &\to H^{k-2}_{\delta_1+\delta_2}(\Reals^3).
\end{aligned}
\end{equation}
See, e.g., \cite{maxwell_rough_2006} Lemma 2.4.

The Euclidean metric $h^{\mathrm{Euc}}_{ab}$ on $\Reals^3$ has 
components $\delta_{ab}$
with respect to standard coordinates. 
A metric $h_{ab}$ on $\Reals^3$ is asymptotically Euclidean 
if it approaches the Euclidean metric $h_{ab}^{\mathrm{Euc}}$
at infinity.
More precisely, we require $h_{ab}-h_{ab}^{\mathrm{Euc}}\in H^k_{\delta}(\Reals^3)$ 
for some $\delta < 0$. Asymptotically Euclidean manifolds 
can be defined more generally \cite{bartnik_mass_1986},
but for simplicity we restrict our attention to initial data on $\Reals^3$.
Nevertheless, the arguments below transfer to the general case
without change.

An \textbf{asymptotically Euclidean conformal fluid seed data set} with regularity $k\ge 3$
and weight parameter $\delta<0$ 
has constant mean curvature $\tau=0$ and the following 
regularity for the remaining seed fields:
\begin{equation}\label{eq:ae-data-set}
\begin{aligned}
h_{ab} \text{ with } (h_{ab}-h^{\mathrm{Euc}}_{ab}) &\in H^k_\delta(\Reals^3),\\
p_{ab} &\in H^{k-1}_{\delta}(\Reals^3),\\
N \text{ with } N>0 \text{ and } N-1 &\in H^{k}_\delta(\Reals^3),\\
T  &\in H^{k-1}_{\beta}(\Reals^3),\\
v^a &\in H^{k-1}_0(\Reals^3),\\
\mathcal DT &\in H^{k-2}_{\beta-1}(\Reals^3),\\
\mathcal A^\Sigma_a  &\in H^{k-2}_{-1}(\Reals^3).
\end{aligned}
\end{equation}
The weight parameter $\beta$ above is defined by $\beta = (\delta-2)/4$ 
and arises so that $T^4\in H^{k-1}_{\delta-2}(\Reals^3)$.
As in the compact case, we assume $p_{ab}$ is trace-free and divergence-free,
that $T\ge 0$, and that $N>0$. An \textbf{asymptotically Euclidean perfect conformal fluid seed data set}
is defined similarly omitting 
the unneeded fields $\mathcal DT$ and $\mathcal A^\Sigma_a$.

We seek a solution of the LCBY equations \eqref{eq:LCBY-Ham-vcf}--\eqref{eq:LCBY-mom-vcf}
of the form $\phi = 1 + u$ with $u\in H^{k}_{\delta}(\Reals^3)$ and $W^a \in H^{k}_{\delta}(\Reals^3)$.
If we find such a solution 
then $\tilde h_{ab} = \phi^4 h_{ab}$ and $\widetilde K_{ab} = \phi^{-2}(p_{ab}+1/(2N)(\ck W)_{ab} )$
together with $\widetilde T=\phi^{-2}T$, $\tilde v^a = \phi^{-2}v^a$, 
$\widetilde{\mathcal DT} = \phi^{-10}\mathcal DT$ and 
$\widetilde {\mathcal A}_a^{\,\Sigma} = A_a^\Sigma +2  \nabla_a \log(\phi)$
solve the constraint equations \eqref{eq:ham-constraint}--\eqref{eq:mom-constraint}.  Moreover, the metric
and second fundamental form are asymptotically Euclidean in the sense that
$\tilde h_{ab}- h_{ab}^{\mathrm{Euc}} \in H^{k}_{\delta}(\Reals^3)$
and $\widetilde K_{ab} \in H^{k-1}_{\delta-1}(\Reals^3)$.

Following the general strategy of Section \ref{secsec:initial-data-compact},
we first consider solvability for inviscid conformal fluids,
in which case equations \eqref{eq:LCBY-Ham-vcf}--\eqref{eq:LCBY-mom-vcf}
become
\begin{align}
    -8\Delta u + R(1+u) -\left|p+\frac{1}{2N}\ck W\right|_h^2(1+u)^{-7} 
    &=2(1+u)^{-3}\mathcal E^{\rm cpf}\label{eq:LCBY-Ham-pcf-ae}\\
    -\nabla^a \left(\frac{1}{2N} (\ck W)_{ab}\right) &=\mathcal J^{{\rm cpf}}_b.\label{eq:LCBY-mom-pcf-ae}
\end{align}

The Yamabe invariant of the asymptotically Euclidean metric $h_{ab}$ defined by
\[
  \mathcal Y_h = \inf_{u \in C^\infty_c(\Reals^3) \atop u\not\equiv 0}  
  \frac{\int_{\Reals^3} 8|\nabla u|_h^2 + Ru^2 dV_h }{||u||_{L^6}^2}
\]
plays a role in the solution theory similar to that in the compact case. 
Since $\mathcal E^{\mathrm{cpf}}\ge 0$, and since the mean curvature vanishes,
the Hamiltonian constraint \eqref{eq:ham-constraint} implies that the scalar
curvature of the solution metric is non-negative.   
In the asymptotically Euclidean setting this implies 
that the Yamabe invariant of the conformal class of $h_{ab}$ 
is necessarily positive, \cite{maxwell_solutions_2005} Proposition 4.1.

\begin{lemma} \label{lem:cpf-data-ae}
    Consider an asymptotically Euclidean perfect fluid seed data set with regularity $k\ge 3$ 
    and weight parameter $-1<\delta<0$.
    If $\mathcal Y_h >0$ then there exists a unique 
     solution $(\phi,W^a$)
     of system  \eqref{eq:LCBY-Ham-pcf-ae}-\eqref{eq:LCBY-mom-pcf-ae}
     with $\phi>0$, $\phi-1\in H^{k}_{\delta}(\Reals^3)$ and $W^a\in H^{k}_{\delta}(\Reals^3)$.
\end{lemma}
\begin{proof}
    Because $\mathcal Y_h=0$, there is a conformal factor $\phi^*=1+u^*$ with $u^*\in H^{k}_{\delta}(\Reals^3)$
    such that $(\phi^*)^4 h_{ab}$ is scalar flat (\cite{maxwell_solutions_2005} Proposition 4.1).  
    Using conformal covariance
    we can therefore assume without loss of generality that $h_{ab}$ is itself scalar-flat.

    Since $-1<\delta<0$, Theorem 4.6 of \cite{maxwell_rough_2006} implies that 
    the vector Laplacian $\nabla^a \left(\ck \cdot\right)_{ab}$ of $h_{ab}$
    is an isomorphism $H^k_\delta(\Reals^3)\to H^{k-2}_{\delta-2}(\Reals^3)$.  
    The same fact remains true for $\nabla^a \left( 1/(2N) \ck \cdot\right)_{ab}$ 
    by arguments parallel to those of Lemma \ref{lem:cpf-data-ae}.
    Because $T^4\in H^{k-1}_{\delta-2}(\Reals^3)$ and because 
    $v^a\in H^{k-1}_{0}(\Reals^3)$, a brief computation shows that 
    $\mathcal J_a^{\mathrm{cpf}}\in H^{k-1}_{\delta-2}(\Reals^3)\subset H^{k-2}_{\delta-2}(\Reals^3)$. Hence there 
    is a unique solution  $W^a\in H^{k}_{\delta}$ of 
    equation \eqref{eq:LCBY-mom-vcf} and we can now focus on the 
    Hamiltonian constraint \eqref{eq:LCBY-Ham-pcf-ae}.

    A solution to equation \eqref{eq:LCBY-Ham-pcf-ae} is obtained by 
    barrier techniques similar to those of Lemma \ref{lem:cpf-data}. 
    Consider the equation
    \begin{equation}\label{eq:lin-lich-ae}
    -8\Delta w = \left|p+\frac{1}{2N}\ck W\right|^2_h + 2\mathcal E^{\mathrm{cpf}}.
    \end{equation}
    Since $p_{ab}$ and $(\ck W)_{ab}$ both belong to $H^{k-1}_{\delta-1}(\Reals^3)$
    it follows that the first term on the right-hand side of equation \eqref{eq:lin-lich-ae}
    belongs to $H^{k-1}_{2\delta-2}(\Reals^3)$.
    Since $T^4\in H^{k-1}_{\delta-2}(\Reals^3)$ and since 
    $v^a\in H^{k-1}_{0}(\Reals^3)$ we find $\mathcal E^{\mathrm{cpf}}\in H^{k-1}_{\delta-2}$.
    Hence the right-hand side of equation \eqref{eq:lin-lich-ae} 
    belongs to $H^{k-1}_{\delta-2}(\Reals^3) \subset H^{k-2}_{\delta-2}(\Reals^3)$.
    Since $-1<\delta<0$,  
    \cite{maxwell_rough_2006} Proposition 5.1 then implies there is a unique solution $w\in H^{k}_\delta(\Reals^3)$ of
    equation \eqref{eq:lin-lich-ae}.  Moreover the maximum principle 
    \cite{maxwell_rough_2006} Lemma 5.2 implies $w\ge 0$.
    
    We claim that $w_-$ and $w_+$ defined by
     $w_-=0$ and $w_+ = w$
    are a sub- and super-solution pair of equation \eqref{eq:LCBY-Ham-pcf-ae}
    respectively.  Indeed, using the fact that $R\equiv 0$, it is obvious
    that $w_-$ is a subsolution.  Moreover,
    \begin{multline*}
        -8\Delta w_+ - \left|p+\frac{1}{2N}\ck W\right|_h(1+w_+)^{-7} -
        2\mathcal E^{\mathrm{cpf}}(1+w_+)^{-3} =\\
        \left|p+\frac{1}{2N}\ck W\right|_h\left( 1 - (1+w)^{-7}\right)
        + 2\mathcal E^{\mathrm{cpf}}\left( 1 - (1+w)^{-3}\right).
    \end{multline*}
    Since $1+w \ge 1$, we see that $1 - (1+w)^{j}\ge 0$ for
    $j=-3$ and $j=-7$ and we conclude that $w_+$ is indeed a supersolution.
    Proposition 6.2 of \cite{maxwell_rough_2006} now implies there is a solution 
    $u\in H^{k}_\delta(\Reals^3)$ with $w_- \le u \le w_+$ of equation \eqref{eq:LCBY-Ham-pcf-ae}.  
    In particular, $\phi = 1+u \ge 1$ is a positive conformal factor.

    To show uniqueness, suppose $u_1$ and $u_2$ both solve \eqref{eq:LCBY-Ham-pcf-ae}.  
    Conformal covariance lets us assume that $u_2\equiv 0$, so 
    we need only show that $u_1\equiv 0$ as well.
    Because $u_2\equiv 0$, $R= |p+1/(2N) \ck W|_h^2 +2\mathcal E^{\mathrm {cpf}}$. Hence
    $u_1$ solves 
    \[
    -8 \Delta u_1 + 
    \left|p+1/(2N) \ck W\right|_h^2 (1+u_1-(1+u_1)^{-7}) + 2\mathcal E^{\mathrm{cpf}} (1+u_1-(1+u_1)^{-3}).
    \]
    Now $(1+u_1-(1+u_1)^{-7}) = k u_1$ where $k\in L^\infty(\Reals^3)$ is positive and a similar
    claim holds for $(1+u_1-(1+u_1)^{-3})$.  Hence $u_1$ solves
    \[
    -8\Delta u_1 + V u_1 = 0    
    \]
    with a non-negative potential $V\in H^0_{\delta-2}$.  Since $-8\Delta+V$ is an isomorphism 
    $H^2_{\delta}(\Reals^3)\to H^0_{\delta-2}(\Reals^3)$ we conclude $u_1\equiv 0$.
\end{proof}

The existence of asymptotically Euclidean viscous fluid solutions of the constraint
equations now follows from an Implicit Function Theorem argument similar to 
that of Theorem \ref{thm:initial-cpct}.
\begin{theorem}\label{thm:initial-ae}
    Consider an asymptotically Euclidean conformal fluid seed data set with regularity $k\ge 3$
    and decay parameter $-1<\delta<0$.  
    Suppose additionally that $\mathcal Y_h > 0$.
    
    Let $(\phi_{\mathrm{cpf}}=1+u_{\mathrm{cpf}}, W^a_{\mathrm{cpf}})$ be the unique solution 
    of the viscosity-free equations \eqref{eq:LCBY-Ham-pcf-ae}--\eqref{eq:LCBY-mom-pcf-ae}
    provided by Lemma \ref{lem:cpf-data-ae}
    for the corresponding perfect seed data set.
    There is a neighborhood $\Gamma$ of 
    $(u_{\mathrm{cpf}}, W^a_{\mathrm{cpf}})$ 
    in $H^k_{\delta}(\Reals^3)\times H^k_{\delta}(\Reals^3)$ such that
    for any sufficiently small viscosity coefficients 
    $c_\chi$, $c_\eta$, $c_\lambda$
    there exists a unique solution $(\phi=1+u, W^a)$ of system
    \eqref{eq:LCBY-Ham-vcf}--\eqref{eq:LCBY-mom-vcf}
    with $(u, W^a) \in \Gamma$.
\end{theorem}
\begin{proof}
Define $F_H[(c_\chi, c_\eta, c_\lambda), (u, W^a)]$ and
$F_M[(c_\chi, c_\eta, c_\lambda), (u, W^a)]$ as in 
equations \eqref{eq:FH}--\eqref{eq:FM} but replacing
$\phi$ in these expressions with $1+u$.

Let $P= \{ u \in H^{k}_\delta(\Reals^3): u>-1\}$ and let
$Y=\{(u,W^a): u\in P, W^a\in H^{k}_{\delta}(\Reals^3)\}$.
We claim $F=F_H\times F_M$ takes $(c_\chi, c_\eta, c_\lambda)\in\Reals^3$ and
$(u, W^a)\in Y$ to $Z:=H^{k-2}_{\delta-2}(\Reals^3)\times H^{k-2}_{\delta-2}(\Reals^3)$.
Establishing this claim amounts to an extended computation examining each of the terms in
equations \eqref{eq:FH}--\eqref{eq:FM}. We present 
a sample computation for two of the more interesting terms
to demonstrate
the needed techniques and leave the remainder for the reader.

Lemma \ref{lem:one-over} below shows that 
$u\mapsto (1+u)^{-1} - 1$ is 
a continuously differentiable map from $P$
to $H^{k}_\delta(\Reals^3)$.  From this and the fact that $H^k_{\delta}(\Reals^3)$
is an algebra it follows that $u\mapsto (1+u)^{-j} - 1$ is 
also a continuously differentiable map $P\to H^{k}_\delta(\Reals^3)$
for each $j\in\Nats$. Now consider the term
\[
\mathcal E^{\mathrm{\chi}} (1+u)^{-3} = \mathcal E^{\chi} + 
\mathcal E^{\mathrm{\chi}} ((1+u)^{-3}-1).
\]
If we establish that $\mathcal E^{\mathrm{\chi}} \in H^{k-2}_{\delta-2}(\Reals^3)$
then the continuity of multiplication 
$H^{k}_{\delta}(\Reals^3)\times H^{k-2}_{\delta-2}(\Reals^3) 
\to H^{k-2}_{\delta-2}(\Reals^3)$ implies that
the map $u\mapsto \mathcal E^{\mathrm{\chi}} (1+u)^{-3}$ is continuously 
differentiable from $P$ to $H^{k-2}_{\delta-2}(\Reals^3)$ as needed.

Recall that
\[
\mathcal E^\chi = 3 T^2\mathcal DT (1 + (4/3)|v|_h^2)
\]
and that we have assumed the following regularity: $T\in H^{k-1}_{\beta}(\Reals^3)$,
$\mathcal DT \in H^{k-2}_{\beta-1}(\Reals^3)$ and $v\in H^{k-1}_0(\Reals^3)$.
Using the multiplication properties \eqref{eq:mult} one readily 
shows
\begin{align*}
T^2\mathcal DT &\in H^{k-2}_{3\beta-1}(\Reals^3),\\
|v|_h^2 &\in H^{k-1}_0(\Reals^3), \\
(1+(4/3)|v|_h^2)&\in H^{k-1}_0(\Reals^3).
\end{align*}
Hence $\mathcal E^\chi \in H^{k-2}_{3\beta -1}$ and 
it remains only to show that $3\beta-1\le \delta -2$.  However,
$\beta = (\delta-2)/4$ and $-1 < \delta$.  Hence
\[
    3\beta -1 = \frac{3}{4}\delta - \frac{5}{2} < 
    \frac{3}{4}\delta - \frac{1}4 - 2 < \delta -2.
\]

Next consider the term $\mathcal E^{\eta} (1+u)^{-3}$.  Following 
the argument above, this term is continuously differentiable
from $u\in P$ to $H^{k-2}_{\delta-2}(\Reals^3)$ so long 
as we can show that $\mathcal E^{\eta}\in H^{k-2}_{\delta-2}$.
Recall
\[
\mathcal E^\eta = -6 T^3 \frac{1}{1+2\gamma^2} \sigma^\Sigma_{ab}v^a v^b.
\]
Multiplication arguments similar to the above show 
$T^3 v^a v^b\in H^{k-1}_{3\beta}$.
Moreover, since $v^a \in H^{k-1}_{0}(\Reals^3)$ 
and since 
$\mathcal A^\Sigma_a\in H^{k-2}_{-1}(\Reals^3)$
it follows that 
$\mathcal D^\Sigma_a v_b \in H^{k-2}_{-1}(\Reals^3)$.
As a consequence $\sigma^\Sigma_{ab}\in H^{k-2}_{-1}(\Reals^3)$
as well.  We can rewrite $1+\gamma^2 = 3+|v|_h^2$ 
and Lemma \ref{lem:one-over} below then implies $1/(1+\gamma^2) = 1/3 + w$
for some $w\in H^{k-1}_0(\Reals^3)$.
Multiplying all factors together we find $\mathcal E^{\eta}\in H^{k-2}_{3\beta-1}(\Reals^3)$,
which we have already shown is contained in $H^{k-2}_{\delta-2}(\Reals^3)$.

The remainder of the claim that $F$ maps continuously differentiably into
$H^{k-2}_{\delta-2}(\Reals^3)\times H^{k-2}_{\delta-2}(\Reals^3)$
is proved either by similar techniques, or by easy arguments involving
linear differential operators.  Having established this claim,
we complete the proof with an argument similar to that of 
Theorem \ref{thm:initial-cpct}.  

From conformal covariance we can assume without loss of generality 
that $u_{\mathrm {cpf}}\equiv 0$ and 
hence $F[0, (0, W^a_{\mathrm{cpf}})]=0$. At the point 
$(0,(0,W^a_{\mathrm{cpf}}))\in \Reals^3\times Y$
 the derivative of $F$ with respect
to $u$ and $W^a$ has the block form
\[
DF = \begin{bmatrix}     
-8\Delta + \Lambda & *\\
0 & \mathrm{div} \left(\frac{1}{2N} \ck \cdot \right)
\end{bmatrix}
\]
where, using an argument parallel to that of Theorem \ref{thm:initial-cpct},
\[
\Lambda = 8\left|p+\frac{1}{2N}\ck W\right|_h^2 + 8\mathcal E^{\mathrm{cpf}}.
\]
In particular $\Lambda\in H^{k-2}_{\delta-2}(\Reals^3)$ and $\Lambda\ge 0$.
Proposition 5.1 of \cite{maxwell_rough_2006} implies the upper left block
$-8\Delta + \Lambda : H^{k}_{\delta}(\Reals^3)\to H^{k-2}_{\delta-2}(\Reals^3)$
is an isomorphism, and the argument of Lemma \ref{lem:cpf-data-ae} 
shows that the lower-right block is also an isomorphism  
$H^{k}_{\delta}(\Reals^3)\to H^{k-2}_{\delta-2}(\Reals^3)$. 
Hence $DF$ is invertible and the 
remainder of the proof using the Implicit Function Theorem now proceeds
exactly as in Theorem \ref{thm:initial-cpct}.
\end{proof}

The following technical lemma completes the argument of Theorem \ref{thm:initial-ae}.
\begin{lemma}\label{lem:one-over}
Let $P=\{u\in H^k_{\delta}(\Reals^3): u>-1\}$ with $k\ge 2$ and $\delta\le 0$.
The map 
\[
u\mapsto \frac{1}{1+u}-1
\]
is continuously differentiable from $P$ to $H^{k}_\delta(\Reals^3)$.
\end{lemma}
\begin{proof}
We first claim that if $u\in P$ then $v:= 1/(1+u)-1=-u/(1+u)\in H^{k}_{\delta}(\Reals^3)$.
Indeed, the claim is proved by direct computation using Sobolev embedding 
as in the proof that $H^k_\delta(\Reals^3)$ is an algebra if $k\ge 2$ and $\delta\le 0$, 
along with the fact
that $1/(1+u)$ is uniformly bounded above and below. 
The main issue is to show that $v$ is a continuously differentiable 
function of $u$.

Define $F:H^{k}_{\delta}(\Reals^3)\times H^{k}_{\delta}(\Reals^3)\to H^{k}_{\delta}(\Reals^3)$
by $F(u,v)=(1+u)(1+v)-1$.  Then $F$ is evidently continuously differentiable.
Consider some $u_0\in P$ and let $v_0 = 1/(1+u_0)-1 = -u_0/(1+u_0)$. Then
$v_0\in H^{k}_\delta(\Reals^3)$
and indeed $v_0\in P$. 

We observe that $F(u_0, v_0)=0$ and that the derivative of $F$ with respect
to $v$ at $(u_0,v_0)$ is $1+u_0$.  Multiplication by $1+u_0$
is a continuous endomorphism of $H^k_{\delta}(\Reals^3)$ and its inverse is
multiplication by $1+v_0$.  Hence the Implicit Function Theorem provides
a neighborhood $U$ of $u_0$ and a continuously differentiable map 
$g:U\to H^k_\delta(\Reals^3)$ with $F(u, g(u))=0$. A brief computation shows
$g(u)=(1+u)^{-1}-1$, which completes the proof.
\end{proof}

\subsection{Application to the Cauchy problem}\label{S:Cauchy}
In this section we show how one can use solutions of the constraints 
obtained in Sections \ref{secsec:initial-data-compact} and \ref{secsec:initial-data-asymptotic}
in conjunction with the results of
\cite{Bemfica-Disconzi-Noronha-2018,Disconzi-2019,Bemfica-Disconzi-Rodriguez-Shao-2021} to solve the Cauchy problem
for the Einstein equations coupled to conformal viscous fluids. 
In these works, the Cauchy problem is solved by gluing together solutions obtained locally in the neighborhood of different points. This is the usual approach to the local Cauchy problem \cite{Choquet-Bruhat-Book-2009,Ringstrom-Book-2009}. Thus, it suffices to assume that we are working locally in a coordinate chart $\{ x^a \}_{a=0}^3$, with $t:=x^0$ identified with a time coordinate and (a portion of) $\Sigma$ identified with $\{t=0\}$. 
Initial data for the Cauchy problem consist of 
\begin{equation}\label{eq:Cauchy-data}
\begin{aligned}
    \left. \left( \epsilon, \partial_t \epsilon, u^a, \partial_t u^a, g_{ab}, \partial_t g_{ab} \right) \right|_{t=0}
\end{aligned}
\end{equation}
and the salient point is that this data is required to satisfy 
the constraint equations.

In view of our results we assume instead that we are given at $t=0$ an initial-data set
\begin{equation}\label{eq:cauchy-data-coordfree}
    (\epsilon, \mathcal{D} \epsilon, v^a, \mathcal{A}^\Sigma_a, h_{ab}, K_{ab})
\end{equation}
for which the constraints are satisfied. From this, we want to determine all of the data \eqref{eq:Cauchy-data}.
Following \cite{Bemfica-Disconzi-Noronha-2018,Disconzi-2019,Bemfica-Disconzi-Rodriguez-Shao-2021}, we work in wave coordinates, although other gauges can be used.

We set $\left. g_{ab} \right|_{t=0} = h_{ab}$ and $\left. \partial_t g_{ab} \right|_{t=0} = 2K_{ab}$ for $a,b=1,2,3$. The remaining components of the metric are gauge choices and we choose $\left. g_{00} \right|_{t=0} = -1$, 
$\left. g_{0a} \right|_{t=0} = 0$, and $\left. \partial_t g_{0a} \right|_{t=0}$ such that $\{ x^a \}_{a=0}^3$ are wave coordinates at $\{t=0\}$. 
As a consequence, all of the Christoffel symbols of the metric are known 
at $t=0$ and the normal to $\{t=0\}$ is $N^a = (\partial_t)^a$.
We set $\epsilon|_{t=0}=\epsilon$ and $u^a|_{t=0}=\gamma N^a + v^a$  with
$\gamma = (1+|v|_h^2)$.  Hence we need only determine the normal derivatives $\partial_t \epsilon$ and $\partial_t u^a$.  

We first observe that the spacetime Weyl connection form $\mathcal A_a$ is 
known at $t=0$.  Indeed, the tangential components are directly specified via
$\mathcal A_a^\Sigma$ and Lemma \ref{lem:weyl-K} shows that $\mathcal A_a N^a$
is computable from the initial data \eqref{eq:cauchy-data-coordfree}.  A
brief computation from the definition of the Weyl derivative \eqref{eq:weyl-def}
and the identity $u^a = \gamma N^a + v^a$ shows
\[
\left. \partial_t \epsilon \right|_{t=0} = \frac{1}{\gamma}\left[\mathcal D\epsilon - v^a \mathcal D_a^\Sigma \epsilon \right] - 4 \mathcal A_a N^a
\]
and is therefore computable from the data \eqref{eq:cauchy-data-coordfree}.

Turning to the normal derivative of $u^a$, we observe that
since the Christoffel symbols of $g_{ab}$ are known at $t=0$, 
we can compute $\partial_t u^a$ directly from $N^b \nabla_b u^a$.
We show in equation \eqref{eq:normal-of-u} below that this latter expression 
is computable from the data \eqref{eq:cauchy-data-coordfree}.

Recall $\mathcal A_a = u^b \nabla_b u^a - (1/3) \nabla_b u^b u_a$ from equation \eqref{eq:A-def}. Since $(\nabla_b u^b u_a)u^a = 0$ we find 
$u^a \mathcal A_a = (1/3) \nabla_b u^b$ and consequently
\[
u^b \nabla_b u_a = \mathcal A_a + \mathcal A_b u^b u_a.
\]
Using the identity $u^a = \gamma N^a + v^a$ we conclude
\begin{equation}\label{eq:normal-of-u}
N^b \nabla_b u_a = \frac{1}{\gamma}\left[   
    \mathcal A_a + \mathcal A_b u^b u_a - v^b \nabla_b u_a \right]
\end{equation}
which is computable from the data \eqref{eq:cauchy-data-coordfree} 
so long as $v^b \nabla_b u^a$ is. 
But this is a standard computation using the second fundamental form:
\begin{align*}
v^b \nabla_b u^a 
& = v^b \nabla_b (\gamma N^a + v^a)\\
& = v^b(\nabla_{b}^{(h)} \gamma)N^a + v^b \nabla_b^{(h)}v^a + h^{ac}v^bK_{bc} + K_{bc}v^bv^c N^a.
\end{align*}

\bibliographystyle{amsalpha-abbrv}
\bibliography{ViscousConformal,manual,RefsMarcelo}

\end{document}